\documentclass[11pt]{article}
\usepackage{graphicx}
\usepackage{subcaption}
\usepackage{todonotes}
\usepackage{mathtools}
\usepackage{soul}
\usepackage{amsthm}
\usepackage{amsmath,amssymb}
\usepackage{systeme}

\usepackage{gensymb}
\usepackage[ruled,vlined,linesnumbered]{algorithm2e}
\SetKwComment{Comment}{/* }{ */}
\SetKwInOut{Input}{input}\SetKwInOut{Output}{output}

\graphicspath{{./figs/}}
\usepackage[absolute]{textpos}

\usepackage{lineno}

\usepackage{fullpage,cite}
\usepackage[top=0.8in, bottom=0.9in, left=0.9in, right=0.9in]{geometry}

\usepackage[pdftex, plainpages = false, pdfpagelabels,
                 bookmarks=false,
                 bookmarksopen = true,
                 bookmarksnumbered = true,
                 breaklinks = true,
                 linktocpage,
                 pagebackref,
                 colorlinks = true,  
                 linkcolor = blue,
                 urlcolor  = cyan,
                 citecolor = red,
                 anchorcolor = green,
                 hyperindex = true,
                 hyperfigures
                 ]{hyperref}

\def\calD{\mathcal{D}}
\def\calC{\mathcal{C}}
\def\calA{\mathcal{A}}

\def\calG{\mathcal{G}}
\def\poly{\text{poly}}
\newtheorem{lemma}{Lemma}
\newtheorem{corollary}{Corollary}
\newtheorem{definition}{Definition}

\newtheorem{theorem}{Theorem}

\title{Deterministic Distance Selection in $O(n^{4/3})$ Time
}
\author{Haitao Wang\thanks{Kahlert School of Computing,
University of Utah, Salt Lake City, UT 84112, USA. {\tt haitao.wang@utah.edu}}
}
\date{}

\begin{document}


\maketitle

\vspace{-0.3in}

\begin{abstract}
Let \(P\) be a set of \(n\) points in the plane. Given an integer
\(K\) with \(1\le K\le {n\choose 2}\), the distance selection problem
asks for the \(K\)-th smallest distance among all pairwise distances
of the points of \(P\). The problem has been studied
extensively over the past four decades. Recently, Chan and Zheng
(2023) gave a randomized algorithm with \(O(n^{4/3})\) expected
running time, while the best known deterministic algorithm runs in \(O(n^{4/3}\log n)\) time.
In this paper, we present a deterministic
\(O(n^{4/3})\)-time algorithm for the problem,
matching the best known randomized complexity. 


We also consider the more general bichromatic version of the problem,
where two point sets \(A\) and \(B\) are given, with
\(m=|A|\) and \(n=|B|\), and the goal is to find the \(K\)-th
smallest distance among the \(mn\) distances between points of \(A\)
and points of \(B\). For this problem, the best previously known
algorithm runs in
\(O((m\log n+n\log m+m^{2/3}n^{2/3})\log(m+n))\) time. We present a
new deterministic algorithm with running time
\(O(m\log^2 n+n\log^2 m+m^{2/3}n^{2/3})\).
\end{abstract}


\section{Introduction}
\label{sec:intro}

Let \(P\) be a set of \(n\) points in the plane. Given an
integer \(K\) with \(1\le K\le {n\choose 2}\), the \emph{distance
selection problem} asks for the \(K\)-th smallest distance among all pairwise distances of the points of \(P\).

Distance selection is a fundamental problem
in computational geometry. Besides its intrinsic interest, it serves as
a prototype for a broad class of geometric optimization problems whose
optimal objective values are realized by interpoint distances~\cite{ref:AgarwalOn24,ref:AvrahamTh15,ref:WangRe23,ref:ChanFa25,ref:KaplanTh23,ref:KatzNe24}. One may
even regard it as a geometric analogue of the classical selection
problem~\cite{ref:BlumTi73}. The problem can be solved trivially in
\(O(n^2)\) time. The first subquadratic algorithm was given by
Chazelle~\cite{ref:ChazelleNe85}. Later, Agarwal, Aronov, Sharir, and
Suri~\cite{ref:AgarwalSe93} obtained an
\(O(n^{3/2}\log^{5/2}n)\)-time algorithm, which was subsequently
improved by Goodrich~\cite{ref:GoodrichGe93} to
\(O(n^{4/3}\log^{8/3}n)\) time. Katz and
Sharir~\cite{ref:KatzAn97} further improved the running time to
\(O(n^{4/3}\log^2 n)\). Wang and
Zhao recently reduced it to \(O(n^{4/3}\log n)\) in ESA 2023~\cite{ref:WangIm25}.

All of the above algorithms are deterministic. Several randomized
algorithms are also known. The randomized algorithm of
Agarwal, Aronov, Sharir, and Suri~\cite{ref:AgarwalSe93} runs in
\(O(n^{4/3}\log^{8/3}n)\) expected time. Matou\v{s}ek~\cite{ref:MatousekRa91} later improved the bound to
\(O(n^{4/3}\log^{2/3}n)\) expected time. Most recently, Chan and
Zheng~\cite{ref:ChanHo23} obtained a randomized algorithm of
\(O(n^{4/3})\) expected time. Their algorithm uses random
sampling and is reminiscent of the randomized linear-time selection
algorithm of Floyd and Rivest~\cite{ref:FloydEx75}. The complexity can
also be expressed as a function of \(K\). In particular, combining
Chan's randomized optimization framework~\cite{ref:ChanOn01} with the
result of Wang~\cite{ref:WangUn23} yields an
\(O(n\log n+n^{2/3}K^{1/3}\log^{5/3}n)\)-time randomized algorithm,
which is particularly attractive when \(K\) is relatively small.

Like Hopcroft's problem and several related geometric range searching
problems~\cite{ref:ChanHo23,ref:EricksonNe96}, it is possible that distance selection has complexity \(\Theta(n^{4/3})\), although no
matching lower bound is currently known. In particular, many problems
whose best known upper bounds are \(O(n^{4/3})\) admit
\(\Omega(n^{4/3})\) lower bounds, at least in a restricted partition-based computational model~\cite{ref:EricksonNe96}, such as Hopcroft's problem~\cite{ref:ChanHo23}, segment intersection counting~\cite{ref:ChanHo23}, computing closest points for segments~\cite{ref:WangAl26,ref:BespamyatnikhCo03}, computing many faces in line arrangement~\cite{ref:WangAn26,ref:AgarwalPa902,ref:AgarwalCo98,ref:EdelsbrunnerTh90,ref:ClarksonCo90,ref:WangCo23}, etc. 
If a similar lower bound could be established for
distance selection, then the randomized algorithm of Chan and
Zheng~\cite{ref:ChanHo23} would be optimal. On the deterministic side,
however, the best known complexity remained
\(O(n^{4/3}\log n)\)~\cite{ref:WangIm25}, leaving a logarithmic gap
between deterministic and randomized algorithms.

\paragraph{Our result.}

In this paper, we close the gap between the best known deterministic
and randomized complexities by presenting a deterministic
\(O(n^{4/3})\)-time algorithm for distance selection.

We consider a more general bichromatic version of the problem. Let
\(A\) and \(B\) be two sets of points in the plane with
\(m=|A|\) and \(n=|B|\), respectively. Let \(S\) denote the set of
all pairwise distances between points of \(A\) and points of \(B\).
Given an integer \(K\) with \(1\le K\le mn\), the goal is to compute
the \(K\)-th smallest distance in \(S\). The original monochromatic
problem can be easily reduced to this bichromatic problem.

We first consider the symmetric case \(m=n\). Previous work typically
solved the bichromatic problem first (with the same time complexities as discussed above), with the monochromatic problem
obtained as a consequence. We follow the same approach. Our main result
is a deterministic \(O(n^{4/3})\)-time algorithm for the symmetric
bichromatic distance selection problem.

We also study the asymmetric case \(m\neq n\). For this problem, the
best previously known algorithm runs in
\(
O(
(m\log n+n\log m+m^{2/3}n^{2/3})
\log(m+n)
)
\)
time~\cite{ref:WangIm25}. By extending the framework developed for the
symmetric case, we obtain a deterministic algorithm with running time
\(
O(m\log^2 n+n\log^2 m+m^{2/3}n^{2/3}).
\)
Note that if \(m=O(1)\), then the running time becomes \(O(n)\),
matching the complexity of the classical linear-time selection
algorithm. On the other hand, if \(m=n\), then the bound becomes
\(O(n^{4/3})\), matching the complexity of our symmetric case
algorithm. 
Thus, the asymmetric bound naturally interpolates between linear time and the symmetric \(O(n^{4/3})\) time.

\subsection{Our techniques}

We give an overview of our algorithm on solving the symmetric case bichromatic problem. 

Like the previous deterministic algorithms~\cite{ref:KatzAn97,ref:WangIm25}, our algorithm follows a
prune-and-search framework. Let \(\delta^*\) denote the sought
\(K\)-th smallest distance in \(S\). The algorithm maintains an
interval \(I\) containing \(\delta^*\). In each iteration, the goal is
to shrink \(I\) to a smaller interval \(I'\subseteq I\) such that
\(\delta^*\in I'\) and
\(
|S\cap I'|
\le
\rho\cdot |S\cap I|
\)
for some constant \(\rho<1\). Repeating this process for $O(\log n)$ iterations reduces
the number of candidate distances to a manageable size. To implement
the framework efficiently, two bottlenecks must be overcome.

\paragraph{Two bottlenecks.}

The first is a \emph{decision bottleneck}. Given a value \(\delta\), we
need to determine whether \(\delta<\delta^*\). 
Observe that $\delta< \delta^*$ if and only if $K_{\delta}< K$, where $K_{\delta}$ is the number of distances of $S$ at most  $\delta$. Hence, we need a {\em decision algorithm} to compute $K_{\delta}$. 

The second is a \emph{representation bottleneck}. Given an interval
\(I\), we need a compact representation of the distances in
\(S\cap I\) so that we can select a small set of representative
distances. These representative distances are then used to shrink
\(I\) to a smaller interval \(I'\). Conceptually, this step plays the
role of a deterministic sampling procedure.

In the algorithm of Katz and Sharir~\cite{ref:KatzAn97}, both
bottlenecks were handled in \(O(n^{4/3}\log n)\) time, leading to an
overall running time of \(O(n^{4/3}\log^2 n)\).

Wang and Zhao~\cite{ref:WangIm25} improved the framework by using a
recent counting algorithm of Wang~\cite{ref:WangUn23}, which computes $K_{\delta}$ in \(O(n^{4/3})\) time for any given $\delta$.
This matches an \(\Omega(n^{4/3})\) lower bound in a restricted model~\cite{ref:ChanHo23,ref:EricksonNe96}, and therefore essentially
settles the decision bottleneck. To address the representation
bottleneck, Wang and Zhao introduced a new partial biclique partition
technique (called \emph{partial batched range searching} in~\cite{ref:WangIm25}). Their technique produces a compact
representation of most, but not all, distances in \(S\cap I\), while
keeping the representation error sufficiently small for the
prune-and-search framework to remain correct. As a result, each
iteration runs in \(O(n^{4/3})\) time and the overall running time
becomes \(O(n^{4/3}\log n)\).

\paragraph{Our solution to the bottlenecks.}

To obtain an \(O(n^{4/3})\)-time algorithm, we address the two
bottlenecks in different ways.

For the decision bottleneck, the lower bound for exact counting
essentially blocks any further improvement. Since a prune-and-search
algorithm requires \(O(\log n)\) iterations and each iteration appears
to require at least one decision query, it is difficult to avoid the
\(O(n^{4/3}\log n)\) barrier using exact counting. In contrast, the
randomized algorithm of Chan and Zheng~\cite{ref:ChanHo23} uses random
sampling and therefore needs only \(O(1)\) counting queries, leading to
an \(O(n^{4/3})\) expected-time algorithm.

The key observation behind our deterministic algorithm is that the
prune-and-search framework does not require the exact number of
distances in \([0,\delta]\). An approximation is sufficient, provided
that the approximation error is carefully controlled. We prove a new
\emph{approximate counting theorem}
(Theorem~\ref{thm:approx-count}) that computes an approximate count for $K_{\delta}$ in
\(o(n^{4/3})\) time. Moreover, the running time decreases as the
allowed error increases. The theorem immediately yields a
\emph{quasi-decision} procedure
(Theorem~\ref{thm:quasidecision}), which replaces the exact decision
algorithm of Wang~\cite{ref:WangUn23}, and the procedure runs in $o(n^{4/3})$ time.
This does not contradict
the $\Omega(n^2)$ lower bound for exact counting, because the quasi-decision
procedure is allowed to return ``unsure'' when the threshold is too
close to the true count.

For the representation bottleneck, we further enhance the partial
biclique partition technique of Wang and Zhao~\cite{ref:WangIm25} and
make it significantly more efficient. Two new ingredients are
introduced. First, we use the \(\Gamma\)-framework of Chan and Zheng~\cite{ref:ChanHo23} to accelerate point location operations once the
subproblem sizes become sufficiently small. Second, we introduce a new
notion called \emph{capture cells} in hierarchical cuttings. Capture
cells allow us to represent only a carefully selected subset of the distances while still guaranteeing that a constant fraction of
all relevant distances are recorded. The idea of using capture cells is crucial to the success of our approach and is one of the most interesting components of this paper. This leads to a new {\em distance
representation theorem} (Theorem~\ref{thm:distance-representation}) that runs in \(o(n^{4/3})\) time.

\paragraph{Putting the pieces together.}

For a positive integer $i$, let $\log^{(i)} n$ denote the $i$-fold iterated logarithm of $n$, i.e.,
$\log^{(1)}n=\log n$, and $\log^{(i)}n=\log(\log^{(i-1)}n)$ for $i\ge 2$.

Our algorithm proceeds in \(O(\log^* n)\) phases and maintains an interval \(I\) containing \(\delta^*\).
For \(i=1,2,\ldots,O(\log^*n)\), the goal of the \(i\)-th phase is to
transform an interval \(I\) satisfying
\[
|S\cap I|
=
O\!\left(
n^{4/3}(\log^{(i-1)}n)^4
\right)
\]
(and \(O(n^2)\) when \(i=1\)) into a smaller interval \(I'\subseteq I\) satisfying
\[
|S\cap I'|
=
O\!\left(
n^{4/3}(\log^{(i)}n)^4
\right),
\]
while maintaining \(\delta^*\in I'\). Initially $i=1$ and
\(I=(0,+\infty)\), and hence \(|S\cap I|=O(n^2)\).

The \(i\)-th phase consists of \(O(\log^{(i)}n)\) iterations. Each
iteration, which reduces the remaining candidate distances by a constant fraction,
takes
\(O(n^{4/3}/\log^{(i)}n)\) time for \(i=1,2\), and
\(O(n^{4/3}/(\log^{(i)}n)^2)\) time for \(i\ge3\).
Consequently, the first two phases each require
\(O(n^{4/3})\) time, while the total time of all remaining phases is
\[
O\!\left(
n^{4/3}
\sum_{i=3}^{O(\log^*n)}
\frac{1}{\log^{(i)}n}
\right)
=
O(n^{4/3}).
\]

After \(O(\log^*n)\) phases, we obtain an interval \(I\) satisfying
\(
|S\cap I|=O(n^{4/3})
\)
and containing \(\delta^*\). A final enumeration step then computes
\(\delta^*\) from \(S\cap I\) in additional \(O(n^{4/3})\) time.
To achieve this, we prove a more general result: given any interval
\(I\), all distances in \(S\cap I\) can be enumerated in
\(
O(n^{4/3}+|S\cap I|)
\)
time. The result follows from an {\em implicit biclique partition theorem}
(Theorem~\ref{thm:implicitbiclique}), which computes a compact
representation of all distances in \(S\cap I\) in an implicit manner
within \(O(n^{4/3})\) time. We believe this theorem may be of
independent interest.
For comparison, the algorithms of Katz and Sharir~\cite{ref:KatzAn97} and Wang and Zhao~\cite{ref:WangIm25} both rely on
explicit representations of the relevant distances. Their techniques
require \(O(n^{4/3}\log n)\) and
\(n^{4/3}2^{O(\log^*n)}\) time, respectively, to construct these
representations. In contrast, our representation is maintained
implicitly, which is crucial for acheiving the \(O(n^{4/3})\) running
time.

\paragraph{The asymmetric case.}

To solve the asymmetric case \(m\neq n\), we seek to reduce the
problem to the symmetric case. For several related geometric problems,
such as Hopcroft's problem~\cite{ref:ChanHo23}, the problem of
computing \(K_\delta\)~\cite{ref:WangUn23}, and computing many faces in line arrangements~\cite{ref:WangAn26},
such a reduction is
relatively straightforward, and the symmetric case algorithm can be
used essentially as a black box. For the distance selection problem, however, the
situation appears to be more subtle. In particular, we cannot simply
invoke the symmetric case algorithm as a black box and obtain the
desired running time. Instead, additional ideas are needed to combine
the symmetric subproblems efficiently.

One possible reason is that distance selection lacks the type of
decomposability enjoyed by the above problems. Solving a subproblem on
two subsets \(A'\subseteq A\) and \(B'\subseteq B\) does not
immediately provide useful information for solving the original
problem on \(A\) and \(B\). Consequently, extra care is needed in
order to exploit the symmetric case machinery while acheiving the
desired time complexity.

Besides distance selection, the techniques developed in this paper may
also be useful for other geometric optimization problems involving
interpoint distances.

\paragraph{Outline.}
The remainder of the paper is organized as follows.
Section~\ref{sec:preliminaries} introduces notation and reviews the
partial biclique partition technique from \cite{ref:WangIm25}.
The next three sections focus on the symmetric case \(m=n\).
Section~\ref{sec:quasidecision} develops the approximate counting
theorem (Theorem~\ref{thm:approx-count}) and the resulting quasi-decision algorithm (Theorem~\ref{thm:quasidecision}).
Section~\ref{sec:representation} proves the distance representation
theorem (Theorem~\ref{thm:distance-representation}). Section~\ref{sec:selection} combines these ingredients to
obtain a deterministic \(O(n^{4/3})\)-time distance selection
algorithm.
Section~\ref{sec:asymmetric} extends the framework to the asymmetric
case \(m\neq n\) and presents the
\(O(m\log^2 n+n\log^2 m+m^{2/3}n^{2/3})\)-time algorithm.

\section{Preliminaries}
\label{sec:preliminaries}

For any two points \(a,b\) in the plane, we use \(\|ab\|\) to denote
their Euclidean distance.
Let \(A\) and \(B\) be two sets of points in the plane. Define
\(
S(A,B)=\{\|ab\|\mid a\in A,\ b\in B\}
\)
to be the set of pairwise distances between points of \(A\) and points
of \(B\). We write \(S\) for \(S(A,B)\) when the point sets are clear
from the context.
We assume that no point of \(A\) is co-located with a point of \(B\).
Consequently, \(0\notin S\).

For an interval \(I\), we use \(S\cap I\) to denote the subset of
distances in \(S\) whose values lie in \(I\). Unless stated otherwise,
all intervals considered in this paper are of the form
\((\alpha,\beta]\), that is, they are open on the left endpoint and
closed on the right endpoint.

For any region \(\sigma\) in the plane, we use \(A(\sigma)\) to denote
the subset of points of \(A\) in \(\sigma\), i.e.,
\(A(\sigma)=A\cap\sigma\). We define \(B(\sigma)\) similarly.

\begin{definition}[Biclique]
Let \(A_t\subseteq A\) and \(B_t\subseteq B\).
A \emph{biclique} is a complete bipartite graph with bipartition
\((A_t,B_t)\), i.e., the graph containing all edges between every point
of \(A_t\) and every point of \(B_t\).
We denote the biclique by \(A_t\times B_t\).
The \emph{size} of the biclique is defined as 
\(
|A_t|+|B_t|.
\)
\end{definition}

A biclique \(A_t\times B_t\) compactly represents all distance pairs
between points of \(A_t\) and points of \(B_t\). We therefore often say
that the biclique {\em records} these pairs or their corresponding distances.
For a point \(a\in A_t\) and a point \(b\in B_t\), we call \((a,b)\) a
\emph{point pair} (or simply a \emph{pair}) of the biclique.
Two bicliques \(A_t\times B_t\) and \(A_{t'}\times B_{t'}\) are
\emph{edge-disjoint} if they do not contain the same point pair.

\subsection{Partial biclique partitions}
\label{subsec:pbs}

The main tool used throughout the paper is a partial biclique partition used by Wang and Zhao~\cite{ref:WangIm25} (called
\emph{partial batched range searching} in
\cite{ref:WangIm25}). Since both the approximate counting algorithm and
the distance representation theorem rely heavily on this construction,
we review it and derive the recurrence relations that will be
used later.

Let \(I=(\alpha,\beta]\) be an interval. For every point \(a\in A\),
define \(D_a(I)\) to be the annulus centered at \(a\) with inner radius
\(\alpha\) and outer radius \(\beta\). When the interval \(I\) is clear
from the context, we simply write \(D_a\). We assume that the outer
boundary belongs to \(D_a\) while the inner boundary does not. Hence,
for any point \(b\in B\),
\(
b\in D_a\) if and only if \(
\|ab\|\in(\alpha,\beta].
\)

Let
\(
\mathcal D=\{D_a\mid a\in A\}
\)
be the set of all annuli centered at the points of \(A\).

The purpose of a partial biclique partition is to represent all
point pairs whose distances lie in a given interval by a small
collection of bicliques while allowing a controlled collection of
unresolved pairs.

\begin{definition}[Partial Biclique Partition]
\label{prob:BRS}
Given an interval \(I\), a \emph{partial biclique partition} for
\(A\) and \(B\) consists of two collections of edge-disjoint bicliques
\(
\Gamma
=
\{A_t\times B_t
\mid
A_t\subseteq A,\,
B_t\subseteq B\}
\)
and
\(
\Pi
=
\{A'_s\times B'_s
\mid
A'_s\subseteq A,\,
B'_s\subseteq B\}
\)
such that the following conditions hold.

\begin{enumerate}
\item
For every pair \((a,b)\in A_t\times B_t\in\Gamma\),
$\|ab\|\in I$.

\item
For every pair of points \(a\in A\) and \(b\in B\) with
\(\|ab\|\in I\), there exists a biclique in either
\(\Gamma\) or \(\Pi\) that contains \((a,b)\).
\end{enumerate}
\end{definition}

In other words, the two collections \(\Gamma\) and \(\Pi\) together
record all pairs \((a,b)\) with distance in \(I\). Every pair recorded
in a biclique of \(\Gamma\) is guaranteed to have distance in \(I\),
while this may not be true for the pairs recorded in \(\Pi\). For this
reason, we call the pairs recorded in \(\Pi\) \emph{uncertain pairs},
and the pairs recorded in \(\Gamma\) \emph{certified pairs}.

The following lemma is a useful tool inherited from \cite{ref:WangIm25}. 

\begin{lemma}[Lemma~1 of \cite{ref:WangIm25}]
\label{lem:partial-brs}
Let \(m=|A|\) and \(n=|B|\). For any parameter
$1\le r\le \min\{m^{1/3},n^{1/3}\}$,
one can compute a partial biclique partition for \(A\) and \(B\) with
collections \(\Gamma\) and \(\Pi\) in
$O(mr\log r+nr)$ time such that:
\begin{enumerate}
\item
$\sum_t(|A_t|+|B_t|)=O(mr\log r+nr)$.

\item
For every biclique \(A'_s\times B'_s\in\Pi\),
$|A'_s| \le {m}/{r^3}$ and $|B'_s|\le {n}/{r^3}$.

\item
The total number of uncertain pairs recorded in \(\Pi\) is $O(mn/r^2)$.
\end{enumerate}
\end{lemma}

We briefly review the algorithm of
Lemma~\ref{lem:partial-brs}, since it will be used repeatedly
throughout the paper. Several concepts and notation introduced here
will also be needed later.

\paragraph{Cuttings.}

Recall that \(\mathcal D\) consists of \(m\) annuli centered at the
points of \(A\), each bounded by two circles. For any compact region
\(\sigma\) in the plane, let \(\mathcal D_\sigma\subseteq\mathcal D\)
denote the subset of annuli whose bounding circles cross the relative interior
of \(\sigma\).

For a parameter \(r\) with \(1\le r\le m\), a
\emph{\((1/r)\)-cutting} for \(\mathcal D\) is a collection \(\Xi\) of
constant-complexity cells with pairwise disjoint interiors whose union
covers the entire plane such that
\(|\mathcal D_\sigma|\le m/r\) for every cell \(\sigma\in\Xi\).
The set \(\mathcal D_\sigma\) is often called the
\emph{conflict list} of \(\sigma\).
The \emph{size} of \(\Xi\) is the number of its cells.

We say that a cutting \(\Xi'\) \emph{\(c\)-refines} a cutting \(\Xi\)
if every cell of \(\Xi'\) is contained in a cell of \(\Xi\) and every
cell of \(\Xi\) contains at most \(c\) cells of \(\Xi'\).

A \emph{hierarchical \((1/r)\)-cutting} for \(\mathcal D\) is a
sequence of cuttings
\(\Xi_0,\Xi_1,\ldots,\Xi_k\) such that:

\begin{enumerate}
\item \(\Xi_i\) \(c\)-refines \(\Xi_{i-1}\) for every
      \(1\le i\le k\), for some constant \(c\);

\item for each \(1\le i\le k\), \(\Xi_i\) is a
      \((1/\rho^i)\)-cutting of size \(O(\rho^{2i})\), for some
      constant \(\rho\);

\item \(\Xi_0\) consists of a single cell equal to the entire plane.
\end{enumerate}

Setting \(k=\lceil\log_\rho r\rceil\) makes the last cutting
\(\Xi_k\) a \((1/r)\)-cutting of size \(O(r^2)\), which is also the total size of all cuttings in the hierarchy.

If a cell \(\sigma\in\Xi_i\) is contained in a cell
\(\sigma'\in\Xi_{i-1}\), then \(\sigma'\) is called the
\emph{parent} of \(\sigma\), and \(\sigma\) is called a
\emph{child} of \(\sigma'\).
In this way, the cells form a tree, called the
\emph{cutting tree}, whose root is the unique cell of \(\Xi_0\).
We often use \(\Xi\) to denote the set of all cells in all
cuttings in the hierarchy. We also view \(\Xi\) as the corresponding cutting tree and use
the notions of ancestors and descendants. 

A hierarchical \((1/r)\)-cutting for \(\mathcal D\) can be computed in
\(O(mr)\) time for any \(1\le r\le m\), e.g., by the algorithm of
Wang~\cite{ref:WangUn23}, which adapts Chazelle's construction~\cite{ref:ChazelleCu93}. The algorithm also computes the conflict list
\(\mathcal D_\sigma\) for every cell \(\sigma\in\Xi\). Consequently,
the total size of all conflict lists is \(O(mr)\).


\paragraph{Algorithm for Lemma~\ref{lem:partial-brs}.}
We briefly review the algorithm; see \cite{ref:WangIm25} for details.

We first construct a hierarchical \((1/r)\)-cutting
\(\{\Xi_0,\Xi_1,\ldots,\Xi_k\}\) for \(\mathcal D\), which takes
\(O(mr)\) time. We use \(\Xi\) to denote the set of all cutting cells.

For each cell \(\sigma\in\Xi\), we compute a canonical subset
\(B(\sigma)=B\cap\sigma\). This can be done in \(O(n\log r)\) time by
processing the points of \(B\) in a top-down manner through the cutting
tree. Since \(k=\Theta(\log r)\), each point of \(B\) is stored in
\(O(\log r)\) canonical subsets. Hence, the total size of all canonical
subsets \(B(\sigma)\), \(\sigma\in\Xi\), is \(O(n\log r)\).

Next, for each cell \(\sigma\in\Xi\), we compute another canonical
subset \(A_\sigma\subseteq A\). Specifically, a point \(a\in A\) belongs
to \(A_\sigma\) if the annulus \(D_a\) contains \(\sigma\) but does not
contain the parent of \(\sigma\). Since each cell has \(O(1)\)
children, all subsets \(A_\sigma\) can be computed in \(O(mr)\) time in
a top-down manner. The total size of all subsets \(A_\sigma\) is also
\(O(mr)\).


By definition, for every cell \(\sigma\in\Xi\), every pair
\((a,b)\in A_\sigma\times B(\sigma)\) satisfies \(\|ab\|\in I\).
Therefore, we add the biclique \(A_\sigma\times B(\sigma)\) to
\(\Gamma\). The resulting bicliques are edge-disjoint. We have shown
that \(\sum_{\sigma\in\Xi}|A_\sigma|=O(mr)\) and
\(\sum_{\sigma\in\Xi}|B(\sigma)|=O(n\log r)\).

Consider now a cell \(\sigma\) of the last cutting \(\Xi_k\). Since
\(|\mathcal D_\sigma|\le m/r\), define
\(\hat A_\sigma=\{a\in A\mid D_a\in\mathcal D_\sigma\}\). We do not know
whether distances between points of \(\hat A_\sigma\) and points of
\(B(\sigma)\) lie in \(I\). If \(|B(\sigma)|>n/r^2\), then we partition
\(B(\sigma)\) into subsets of size between \(n/(2r^2)\) and \(n/r^2\),
called \emph{standard subsets}. Since \(|B|=n\) and \(\Xi_k\) contains
\(O(r^2)\) cells, the total number of standard subsets over all cells of
\(\Xi_k\) is \(O(r^2)\).

For each standard subset \(\hat B(\sigma)\subseteq B(\sigma)\), we form
a recursive subproblem \((\hat A_\sigma,\hat B(\sigma))\). There are
\(O(r^2)\) such subproblems. Since
\(|\hat A_\sigma|\le m/r\) and \(|\hat B(\sigma)|\le n/r^2\), the
running time satisfies the recurrence
\begin{equation}
\label{formula:PrimalRecurrence}
T(m,n)
=
O(mr+n\log r)
+
O(r^2)\cdot
T\!\left(\frac{m}{r},\frac{n}{r^2}\right).
\end{equation}

Throughout the paper, we use \(T(m,n)\) interchangeably to denote both
the running time and the corresponding subproblem involving two point
sets of sizes \(m\) and \(n\).

We next solve the problem in a dual setting by switching the roles of
\(A\) and \(B\), i.e., defining annuli centered at the points of \(B\).
Symmetrically, we obtain
\begin{equation}
\label{formula:DualRecurrence}
T(m,n)
=
O(nr+m\log r)
+
O(r^2)\cdot
T\!\left(\frac{m}{r^2},\frac{n}{r}\right).
\end{equation}

Applying \eqref{formula:DualRecurrence} to each subproblem of
\eqref{formula:PrimalRecurrence} with the same parameter \(r\) yields
\begin{equation}
\label{formula:balanced}
T(m,n)
=
O(mr\log r+nr)
+
O(r^4)\cdot
T\!\left(\frac{m}{r^3},\frac{n}{r^3}\right).
\end{equation}

The balanced recurrence implies that, in
\(O(mr\log r+nr)\) time, one can compute a collection of \(O(r^4)\)
edge-disjoint bicliques \(A_t\times B_t\) such that every pair
\((a,b)\in A_t\times B_t\) satisfies \(\|ab\|\in I\). Furthermore,
\(\sum_t(|A_t|+|B_t|)=O(mr\log r+nr)\). We return this collection as
\(\Gamma\).

In addition, we obtain \(O(r^4)\) bicliques
\(A'_s\times B'_s\) corresponding to the subproblems
\(T(m/r^3,n/r^3)\). For these bicliques, we do not know whether
\(\|ab\|\in I\) for pairs \((a,b)\in A'_s\times B'_s\). We return the
collection of all such bicliques as \(\Pi\). Consequently,
\(|A'_s|\le m/r^3\) and \(|B'_s|\le n/r^3\) for every
\(A'_s\times B'_s\in\Pi\), and the total number of uncertain pairs
recorded in \(\Pi\) is \(O(mn/r^2)\). This proves
Lemma~\ref{lem:partial-brs}.

We remark that Recurrence~\eqref{formula:balanced} is the key algorithmic consequence
of the partial biclique partition framework. It will serve as the
starting point for both the approximate counting algorithm in
Section~\ref{sec:quasidecision} and the distance representation
algorithm in Section~\ref{sec:representation}.

\section{Approximate counting and quasi-decision}
\label{sec:quasidecision}

In this section and the next two sections, we assume that \(|A|=|B|=n\). 

This section develops an approximate counting algorithm for
interval distances. The quasi-decision theorem then follows
immediately by comparing the approximate count with a given threshold.

For an interval \(I\), let
\(\kappa(I)=|S\cap I|\) denote the number of distances of \(S\) whose
values lie in \(I\).

Our main goal in this section is to prove the following approximate
counting theorem.

\begin{theorem}[Approximate Counting]
\label{thm:approx-count}
Let \(I\) be any interval. For every integer \(i\) with
\(1\le i\le O(\log^* n)\), one can compute an integer
\(\hat\kappa(I)\) such that
\[
|\kappa(I)-\hat\kappa(I)|
\le
c_0\cdot n^{4/3}(\log^{(i)}n)^4
\]
for some constant \(c_0>0\).

The runtime is \(O(n^{4/3}/\log^{(i)}n)\) for \(i=1,2\), and
\(O(n^{4/3}/(\log^{(i)}n)^2)\) for \(i\ge 3\).
\end{theorem}

\paragraph{Remark.}
Ideally, we would like to compute \(\kappa(I)\) exactly. However,
exact computation requires \(O(n^{4/3})\) time~\cite{ref:WangUn23,ref:ChanHo23}, which is optimal in a restricted computational model. The purpose of
Theorem~\ref{thm:approx-count} is to obtain a faster algorithm by
computing only an approximation of \(\kappa(I)\) while controlling the
approximation error. Furthermore, the theorem provides a trade-off
between running time and error: allowing a larger error leads to a
faster algorithm.
\medskip

Theorem~\ref{thm:approx-count} immediately yields the following
quasi-decision theorem, which will be used later in the distance
selection algorithm of Section~\ref{sec:selection}.

\begin{theorem}[Quasi-Decision]
\label{thm:quasidecision}
Given an interval \(I\) and an integer \(K\) with
\(1\le K\le n^2\), for every \(i\) with
\(1\le i\le O(\log^* n)\), there exists an algorithm whose output has
one of the following three outcomes:

\begin{enumerate}
\item the algorithm reports \(\kappa(I)<K\);

\item the algorithm reports \(\kappa(I)>K\);

\item the algorithm reports ``unsure'', in which case
\[
|\kappa(I)-K|
\le
c_0\cdot n^{4/3}(\log^{(i)}n)^4
\]
for some constant \(c_0>0\).
\end{enumerate}

The algorithm runs in \(O(n^{4/3}/\log^{(i)}n)\) time for \(i=1,2\),
and in \(O(n^{4/3}/(\log^{(i)}n)^2)\) time for \(i\ge 3\).
\end{theorem}

\begin{proof}
Compute \(\hat\kappa(I)\) using
Theorem~\ref{thm:approx-count}. If
\(\hat\kappa(I)+c_0n^{4/3}(\log^{(i)}n)^4<K\), then
\(\kappa(I)<K\). If
\(\hat\kappa(I)-c_0n^{4/3}(\log^{(i)}n)^4>K\), then
\(\kappa(I)>K\). Otherwise, report ``unsure''.
The correctness follows immediately from the error guarantee of
Theorem~\ref{thm:approx-count}.
\end{proof}

In the remainder of this section, we prove
Theorem~\ref{thm:approx-count}. The proof naturally splits into two
cases. For \(i=1,2\), the algorithm follows directly from repeated
applications of Lemma~\ref{lem:partial-brs}. The case \(i\ge3\) is
substantially more difficult and new techniques are required. 

\subsection{The cases \texorpdfstring{$i=1,2$}{i=1,2}}
\label{subsec:quasi-small}

For \(i=1,2\), the approximate counting algorithm follows directly from
the balanced recurrence~\eqref{formula:balanced}. For convenience, we
restate the recurrence here in the symmetric case \(m=n\):
\begin{equation}
\label{equ:small}
T(n,n)
=
O(nr\log r)
+
O(r^4)\cdot
T\!\left(
\frac{n}{r^3},
\frac{n}{r^3}
\right).
\end{equation}

The algorithm computes an approximate count \(\hat\kappa(I)\) by
recursively applying Lemma~\ref{lem:partial-brs}. At each recursive
step, we add to \(\hat\kappa(I)\) the number of certified pairs
recorded by the bicliques in \(\Gamma\). Since every certified pair has
distance in \(I\), the only source of error comes from the uncertain
pairs represented by \(\Pi\). Consequently,
\(|\hat\kappa(I)-\kappa(I)|\) is bounded by the total number of
uncertain pairs in the remaining recursive subproblems.

For \(i=1\), we apply Lemma~\ref{lem:partial-brs} once with
\(r=n^{1/3}/\log^2 n\). The runtime is \(O(nr\log r)\), which is
\(O(n^{4/3}/\log n)\), and the number of uncertain pairs is
\(O(n^2/r^2)=O(n^{4/3}\log^4 n)\). This proves
Theorem~\ref{thm:approx-count} for \(i=1\).

For \(i=2\), we apply Lemma~\ref{lem:partial-brs} recursively twice.
The first application again uses
\(r=n^{1/3}/\log^2 n\). For each biclique
\(A'_s\times B'_s\in\Pi\), we apply
Lemma~\ref{lem:partial-brs} again with
\(r=\log^2 n/\log^2\log n\). The total runtime is
\(O(n^{4/3}/\log\log n)\), while the total number of uncertain pairs
after the second recursive step is
\(O(n^{4/3}/\log^4\log n)\). This proves
Theorem~\ref{thm:approx-count} for \(i=2\).

\subsection{The case \texorpdfstring{$i\geq 3$}{i>=3}}
\label{subsec:quasi-large}

For \(i\ge 3\), one might try to continue applying
Recurrence~\eqref{equ:small}. However, every application of the
recurrence introduces a constant-factor loss. After \(i\) recursive
levels, both the running time and the error bound would incur an extra
factor \(2^{O(i)}\). Consequently, this approach is sufficient only
when \(i=O(1)\). To obtain the bounds stated in
Theorem~\ref{thm:approx-count} for general \(i\), a different approach
is needed.

The key observation is that for approximate counting we do not need to
construct all bicliques explicitly. Instead, it is sufficient to know,
for each cell of the last cutting, how many annuli completely contain
the cell. This allows us to count certified pairs implicitly through
cutting cells rather than explicitly through bicliques.

The remaining difficulty is point location. To compute the contribution
of a cell, we need to determine, for every point of \(B\), the cell of
the last cutting that contains it. A straightforward implementation
would require logarithmic time per point. To avoid this overhead, we
use the \(\Gamma\)-framework of Chan and Zheng~\cite{ref:ChanHo23}. The framework becomes applicable because, after
two levels of recursion, the remaining subproblem size is already
extremely small. The details are given below.

We consider a general \(i\ge 3\).

We first apply Recurrence~\eqref{equ:small} twice with
\(r=n^{1/3}/\log^2 n\), exactly as in the case \(i=2\). This takes
\(O(n^{4/3}/\log\log n)\) time and produces
\(O(n^{4/3}/\log^8 n)\) certified bicliques together with
\(O(n^{4/3}/\log^8 n)\) uncertain bicliques. We add
\(\sum_t (|A_t|\cdot |B_t|)\) over all certified bicliques to
\(\hat\kappa(I)\) (initially, \(\hat\kappa(I)=0\)). Note that the total
size of all certified bicliques satisfies
\(\sum_t(|A_t|+|B_t|)=O(n^{4/3}/\log\log n)\).

After two recursive levels, each uncertain biclique corresponds to a
subproblem \(T(q,q)\), where
\(q=\log^6\log n\). Hence, the total running time satisfies
\begin{equation}
\label{equ:tiny00}
T(n,n)
=
O(n^{4/3}/\log\log n)
+
O(n^{4/3}/\log^8 \log n)\cdot T(q,q).
\end{equation}

We next solve the subproblem \(T(q,q)\). Since \(q\) is sufficiently
small (e.g., \(2^{\poly(q)}=O(n)\) for any polynomial function
\(\poly(\cdot)\)), we can exploit the \(\Gamma\)-framework of
Chan and Zheng~\cite{ref:ChanHo23}.

Let \(g=\log^{(i-2)}q\). We first apply
Recurrence~\eqref{formula:DualRecurrence} with \(m=n=q\) and
\(r=q^{1/3}/g^2\), obtaining
\begin{equation}
\label{equ:tiny10}
T(q,q)
=
O\!\left(\frac{q^{4/3}}{g^2}\right)
+
O\!\left(\frac{q^{2/3}}{g^4}\right)
\cdot
T\!\left(
q^{1/3}g^4,
q^{2/3}g^2
\right).
\end{equation}

As before, we add the number of certified pairs produced by the
recurrence to \(\hat\kappa(I)\).

\paragraph{Solving subproblems $\boldsymbol{T(q^{1/3}g^4,q^{2/3}g^2)}$.}
We now solve the subproblem
\(T(q^{1/3}g^4,q^{2/3}g^2)\), which corresponds to an uncertain
biclique \(A'\times B'\) with $A'\subseteq A$, $B'\subseteq B$, 
\(|A'|\le q^{1/3}g^4\), and \(|B'|\le q^{2/3}g^2\).
Let \(m'=|A'|\) and \(n'=|B'|\).

Let \(\mathcal D'\) be the set of annuli defined by the points of $A'$ with respect to the interval $I$. We compute a hierarchical
\((1/r')\)-cutting $\{\Xi_0, \Xi_1, ..., \Xi_k\}$ for \(\mathcal D'\), where
\(r'=m'/g^6\). This takes \(O(m'r')\) time. For every cell \(\sigma\)
of the last cutting $\Xi_k$, let \(m'_\sigma\) denote the number of annuli of
\(\mathcal D'\) that completely contain \(\sigma\). Since each cutting
cell has \(O(1)\) children, all values \(m'_\sigma\) can be computed in
\(O(m'r')\) time by traversing the cutting tree.

Next, we locate every point of \(B'\) in the last cutting $\Xi_k$.
A straightforward point location procedure would require
\(O(\log r')\) time per point, which is not efficient enough for our
purpose. To avoid this overhead, we apply the \(\Gamma\)-framework of
Chan and Zheng~\cite{ref:ChanHo23}. 

Since \(q=\log^6\log n\), the subproblem size is sufficiently small for
the framework to be effective. In a manner similar to
\cite[Lemma~4.4]{ref:ChanHo23}, each point location can be performed in
amortized \(O(1)\) time (say, after $O(n^{5/4})$-time preprocessing for $A$ and $B$). Note that the point
location problem in \cite[Lemma~4.4]{ref:ChanHo23} is on a line
arrangement, whereas our point location problem is on a cutting
\(\Xi_k\). Nevertheless, the same idea applies by exploiting the
hierarchical structure of the cuttings; see Section~5.3 of the arXiv
version of \cite{ref:WangAn26} for a similar problem.

Consequently, all point locations for the points of \(B'\) can be
performed in \(O(n')\) amortized time. After that, we obtain the
canonical subsets \(B'(\sigma)=B'\cap \sigma\) for all cells
\(\sigma\in\Xi_k\).

For each cell \(\sigma\), we add
\(m'_\sigma\cdot |B'(\sigma)|\) to \(\hat\kappa(I)\). This is the key
difference from the cases \(i=1,2\): certified pairs are counted
implicitly through cutting cells rather than explicitly through
bicliques.

The unresolved pairs are represented by bicliques
\(\hat A'_\sigma\times B'(\sigma)\), where
\(\hat A'_\sigma=\{a\in A'\mid D_a\in\mathcal D'_\sigma\}\).
Since \(|\hat A'_\sigma|\le m'/r'\) and
\(\sum_\sigma |B'(\sigma)|=n'\), the total number of unresolved pairs
for the subproblem $T(m',n')$ is \(O(m'n'/r')\). The running time of the
subproblem is \(O(m'r'+n')\).

Since \(r'=m'/g^6\), \(m'\le q^{1/3}g^4\), and \(n'\le q^{2/3}g^2\), we obtain
\(T(m',n')=O(q^{2/3}g^2)\), while the number of unresolved pairs is
\(O(q^{2/3}g^8)\). Substituting these bounds into
\eqref{equ:tiny10} yields
\(T(q,q)=O(q^{4/3}/g^2)\), and the total number of unresolved pairs is
\(O(q^{4/3}g^4)\).

Substituting these bounds into \eqref{equ:tiny00} yields
\(T(n,n)=O(n^{4/3}/g^2)\), and the total number of unresolved pairs is
\(O(n^{4/3}g^4)\). Recalling that
\(g=\log^{(i-2)}q\), we conclude that the total running time is
\(O(n^{4/3}/(\log^{(i)}n)^2)\), while
\(
|\hat\kappa(I)-\kappa(I)|
=
O(n^{4/3}(\log^{(i)}n)^4).
\)

This proves Theorem~\ref{thm:approx-count} for all \(i\ge 3\).

\section{Distance representation}
\label{sec:representation}

The approximate counting algorithm of
Section~\ref{sec:quasidecision} is sufficient for quasi-decision.
However, it does not directly yield an \(O(n^{4/3})\)-time distance
selection algorithm, because distance selection ultimately requires a
compact representation of the candidate distances themselves.
This section develops such a representation. The resulting theorem is
the main structural result of the paper and constitutes the principal
new ingredient that enables the deterministic \(O(n^{4/3})\)-time
algorithm.

Given an interval \(I\), our goal is to represent the distances of
\(S\cap I\) by a collection of bicliques. Ideally, we would like to
represent all distances of \(S\cap I\). However, it is unlikely 
this can be achieved in time \(o(n^{4/3})\). Instead, our strategy is
to represent only a constant fraction of the distances in \(S\cap I\), which is sufficient for our prune-and-search distance selection algorithm as will be shown in Section~\ref{sec:selection}.

This idea is motivated by Lemma~\ref{lem:partial-brs}. The lemma
already produces a compact representation of most distance pairs while
leaving a controlled collection of unresolved pairs. However, the
resulting algorithm is not efficient enough for our purposes. Our new
algorithm goes one step further: among the certified pairs, it records
only a constant fraction of them. This additional flexibility allows us
to obtain a substantially faster representation algorithm.

Like the approximate counting theorem, the distance representation
theorem provides a trade-off between running time and representation
error. Roughly speaking, the more unresolved pairs we allow, the more
efficient the algorithm becomes. For example, if no unresolved pairs
are allowed, our algorithm runs in \(O(n^{4/3})\) time. By comparison,
the technique of Wang and Zhao~\cite[Theorem~1]{ref:WangIm25} requires
\(n^{4/3}\cdot 2^{O(\log^*n)}\) time for the same task.%
\footnote{One may wonder whether the \(\Gamma\)-framework~\cite{ref:ChanHo23} can be used to remove the
\(2^{O(\log^*n)}\) factor. It is not clear to us whether this is
possible. See the discussion in \cite{ref:WangIm25} about this challenge following Theorem~1.}

The following theorem summarizes the main result of this section.

\begin{theorem}[Distance Representation]
\label{thm:distance-representation}
Given an interval \(I\), for every integer \(i\) with
\(1\le i\le O(\log^*n)\) such that
\[
|S\cap I|
\ge
c_1\cdot n^{4/3}(\log^{(i)}n)^4
\]
for a sufficiently large constant \(c_1>0\), one can compute a
collection \(\Gamma=\{A_t\times B_t\}\) of edge-disjoint bicliques such
that:
\begin{enumerate}
\item every pair recorded by a biclique of \(\Gamma\) has distance in
      \(I\);

\item
      \(
      \sum_t |A_t||B_t|
      \ge
      0.4\cdot |S\cap I|;
      \)

\item \(\sum_t(|A_t|+|B_t|)\) is
      \(O(n^{4/3}/\log^{(i)}n)\) for \(i=1,2\), and
      \(O(n^{4/3}/(\log^{(i)}n)^2)\) for \(i\ge 3\).
\end{enumerate}

The construction time is
\(O(n^{4/3}/\log^{(i)}n)\) for \(i=1,2\), and
\(O(n^{4/3}/(\log^{(i)}n)^2)\) for \(i\ge 3\).
\end{theorem}

\paragraph{Remark.}
The second statement implies that the bicliques of \(\Gamma\) record at
least forty percent of the distances in \(S\cap I\). Although the
theorem does not explicitly mention unresolved pairs, they are an
essential part of the construction. 

\medskip

In the remainder of this section, we prove
Theorem~\ref{thm:distance-representation}. The proof follows the same
high-level framework as the approximate counting algorithm. However,
there is a crucial difference. Theorem~\ref{thm:approx-count} only
requires counting target distances, whereas here we must organize them
into bicliques. Consequently, the representation problem is
substantially more challenging.

We begin with the easier cases \(i=1,2\), where the algorithm follows
directly from Lemma~\ref{lem:partial-brs}. The case \(i\ge3\) requires
new ideas and constitutes the main technical contribution of this
section.

\subsection{The cases \texorpdfstring{$i=1,2$}{i=1,2}}

For \(i=1,2\), the distance representation theorem follows directly from
the recursive construction underlying
Theorem~\ref{thm:approx-count}. 

As in Section~\ref{subsec:quasi-small}, every certified biclique
produced by Recurrence~\eqref{equ:small} is added to \(\Gamma\), while the bicliques
corresponding to the remaining recursive subproblems form \(\Pi\).

For \(i=1\), we apply Lemma~\ref{lem:partial-brs} once with
\(r=n^{1/3}/\log^2 n\). The running time is
\(O(n^{4/3}/\log n)\), and
\(\sum_t(|A_t|+|B_t|)=O(n^{4/3}/\log n)\).
The total number of uncertain pairs recorded in \(\Pi\) is
\(O(n^{4/3}\log^4 n)\).

Let \(N(\Pi)\) denote the number of uncertain pairs recorded in
\(\Pi\). Since all certified pairs are represented in \(\Gamma\), we have 
\[
\sum_{A_t\times B_t\in\Gamma}|A_t||B_t|
=
|S\cap I|-N(\Pi).
\]

Since \(N(\Pi)\le c\,n^{4/3}\log^4 n\) for some constant \(c>0\), and
Theorem~\ref{thm:distance-representation} assumes
\(|S\cap I|\ge c_1n^{4/3}\log^4 n\), we obtain
\[
\sum_{A_t\times B_t\in\Gamma}|A_t||B_t|
\ge (1-\frac{c}{c_1})|S\cap I|.
\]

Therefore, by choosing \(c_1\) sufficiently large, we have

\[
\sum_{A_t\times B_t\in\Gamma}|A_t||B_t|
\ge
0.4\,|S\cap I|.
\]

This proves the theorem for \(i=1\).

For \(i=2\), we apply Recurrence~\eqref{equ:small} twice.
As shown in Section~\ref{subsec:quasi-small}, the running time becomes
\(O(n^{4/3}/\log\log n)\), while the total number of uncertain pairs is
\(O(n^{4/3}\log^4\log n)\).
Again, all certified pairs are recorded in \(\Gamma\), and therefore
\(
\sum_{A_t\times B_t\in\Gamma}|A_t||B_t|
=
|S\cap I|-N(\Pi).
\)
Using the assumption
\(|S\cap I|\ge c_1n^{4/3}(\log\log n)^4\) and choosing \(c_1\)
sufficiently large yields
$\sum_{A_t\times B_t\in\Gamma}|A_t||B_t| \ge 0.4\,|S\cap I|$.

The size bound
\(\sum_t(|A_t|+|B_t|)=O(n^{4/3}/\log\log n)\)
follows directly from the recurrence. This proves the theorem for
\(i=2\).

\subsection{The case \texorpdfstring{$i\geq 3$}{i>=3}}
\label{subsec:disrep-large}

For \(i\ge 3\), one might try to continue applying
Recurrence~\eqref{equ:small}. As discussed in
Section~\ref{subsec:quasi-large}, this approach incurs a constant-factor
loss at every recursive level and therefore does not yield the bounds
of Theorem~\ref{thm:distance-representation} for general \(i\).
Consequently, a different approach is required.

The high-level framework remains the same as in
Section~\ref{subsec:quasi-large}. We again apply the recurrence twice
and reduce the problem to a collection of tiny subproblems
\(T(q,q)\), where \(q=\log^6\log n\). However, the way we solve these
subproblems is fundamentally different.

For approximate counting, it suffices to know, for each cell of the
last cutting, how many annuli completely contain the cell. This allows
us to count certified pairs without explicitly constructing bicliques.
For distance representation, however, this information is no longer
sufficient. We must explicitly organize the target pairs into
bicliques, and therefore a new idea is needed.

We follow the same notation as in
Section~\ref{subsec:quasi-large} and obtain
Recurrence~\eqref{equ:tiny00}. We add all certified bicliques produced
during the first two recursive levels to \(\Gamma\). This takes
\(O(n^{4/3}/\log\log n)\) time, and the total contribution to
\(\sum_t(|A_t|+|B_t|)\) is bounded by the same bound.

Next, for every subproblem \(T(q,q)\), we apply
Recurrence~\eqref{equ:tiny10} and again add all certified bicliques to
\(\Gamma\). This contributes \(O(q^{4/3}/g^2)\) time and the same
asymptotic amount to the total representation size.

We now solve the subproblem
\(T(q^{1/3}g^4,q^{2/3}g^2)\) in Recurrence~\eqref{equ:tiny10}, which corresponds to an unresolved
biclique \(A'\times B'\) with \(A'\subseteq A\),
\(B'\subseteq B\), \(|A'|\leq q^{1/3}g^4\), and
\(|B'|\leq q^{2/3}g^2\). Let \(m'=|A'|\) and \(n'=|B'|\).

Let \(\mathcal D'\) be the set of annuli defined by the points of $A'$ with respect to the interval $I$.  We compute a hierarchical
\((1/r')\)-cutting \(\{\Xi_0,\Xi_1,\ldots,\Xi_k\}\) for
\(\mathcal D'\), where \(r'=m'/g^6\). This takes \(O(m'r')\) time.
Let \(\Xi\) denote the set of all cutting cells.

In Section~\ref{subsec:quasi-large}, the subproblem
\(T(q^{1/3}g^4,q^{2/3}g^2)\) was solved by counting, for each cell of
the last cutting, how many annuli completely contain the cell. That
approach is sufficient for approximate counting because the goal is
merely to compute a number.

For distance representation, the situation is different. We must
organize the target pairs into bicliques, and therefore the numbers
computed in Section~\ref{subsec:quasi-large} are no longer useful by
themselves. To overcome this difficulty, we introduce a new notion
called \emph{capture cells}. Capture cells allow us to represent a
constant fraction of the certified pairs while maintaining the desired
running-time bound.

In what follows, we first define capture cells and then show how they
lead to the distance representation theorem.


\subsubsection{Capture cells}

For a cell \(\sigma\in\Xi\), define \(\mathcal A'(\sigma)\) to be the
set of points \(a\in A'\) whose annuli \(D_a\) completely contain
\(\sigma\). We further define \(\hat A'_\sigma\) to be the set of
points \(a\in A'\) whose annuli belong to the conflict list
\(\mathcal D'_\sigma\).

\begin{definition}
    The inequality $|\mathcal A'(\sigma)| \ge |\hat A'_\sigma|$ is called the \emph{capture condition}.
\end{definition}


A cell \(\sigma\) is called a \emph{Type-I capture cell} if it satisfies
the capture condition but its parent does not. More precisely,
\(\sigma\) is a Type-I capture cell if
\(|\mathcal A'(\sigma)|\ge |\hat A'_\sigma|\) and
\(|\mathcal A'(\sigma')|<|\hat A'_{\sigma'}|\), where \(\sigma'\) is
the parent of \(\sigma\).
Note that a cell of \(\Xi_0\) can never satisfy the capture condition.
Indeed, \(\Xi_0\) consists of a single cell equal to the entire plane,
and thus \(|\mathcal A'(\sigma)|=0\) while
\(|\hat A'_\sigma|=m'\).

Recall that the cells of \(\Xi\) form a tree and the cells of the last
cutting \(\Xi_k\) are {\em leaf cells}. We consider a cell to be both an
ancestor and a descendant of itself.

A leaf cell whose root-to-leaf path contains no Type-I capture cell is
called a \emph{Type-II capture cell}.

We refer to both Type-I and Type-II cells as \emph{capture cells}, and
use \(\Sigma\) to denote the set of all capture cells.

\begin{lemma}
\label{lem:unique-type1}
Every root-to-leaf path in \(\Xi\) contains exactly one capture cell.
\end{lemma}
\begin{proof}
The Type-II capture cells guarantee that every root-to-leaf path
contains at least one capture cell. Since only leaf cells can be
Type-II capture cells, it suffices to show that every root-to-leaf path
contains at most one Type-I capture cell.

Let \(\sigma\) be a Type-I capture cell and let \(\sigma'\) be its
parent. By definition,
\(|\mathcal A'(\sigma')|<|\hat A'_{\sigma'}|\).

Consider any ancestor \(\sigma''\) of \(\sigma'\). Since
\(\sigma''\) contains \(\sigma'\), we have
\(\hat A'_{\sigma'}\subseteq \hat A'_{\sigma''}\) and therefore
\(|\hat A'_{\sigma'}|\le |\hat A'_{\sigma''}|\). On the other hand,
\(\mathcal A'(\sigma'')\subseteq \mathcal A'(\sigma')\), and thus
\(|\mathcal A'(\sigma'')|\le |\mathcal A'(\sigma')|\).

Since
\(|\mathcal A'(\sigma')|<|\hat A'_{\sigma'}|\),
it follows that
$|\mathcal A'(\sigma'')|
<
|\hat A'_{\sigma''}|$.
Hence, \(\sigma''\) cannot satisfy the capture condition. Therefore no
ancestor of \(\sigma\) (other than $\sigma$ itself) can be a Type-I capture cell, proving the
lemma.
\end{proof}

Lemma~\ref{lem:unique-type1} immediately implies the following.

\begin{corollary}
\label{coro:disjoint}
Every two capture cells have disjoint interiors.
\end{corollary}
\begin{proof}
By Lemma~\ref{lem:unique-type1}, two capture cells cannot belong to the
same root-to-leaf path. According to the definition of the hierarchical cuttings, 
two cells of $\Xi$ have overlapped interiors only if one cell is the ancestor of the other. Hence, the two capture cells cannot have overlapped interiors.
\end{proof}

For every capture cell \(\sigma\), every pair recorded by the biclique
\(\mathcal A'(\sigma)\times B'(\sigma)\) has distance in \(I\).
Let \(\mathcal C\) denote the collection of all such bicliques.
By Corollary~\ref{coro:disjoint}, no point of \(B'\) belongs to two
 sets \(B'(\sigma)\) for two capture cells. Consequently, the bicliques of
\(\mathcal C\) are edge-disjoint.

We add all bicliques of \(\mathcal C\) to \(\Gamma\). 

Lemma~\ref{lem:bicliquesize} shows that the total size of all bicliques of $\calC$ is $O(m'r'+n')$. 

\begin{lemma}
\label{lem:bicliquesize}
$\sum_{\sigma\in \Sigma}|\calA'(\sigma)|=O(m'r')$ and $\sum_{\sigma\in \Sigma}(|\calA'(\sigma)|+ |B'(\sigma)|)=O(m'r'+n')$.
\end{lemma}
\begin{proof}
First, since the sets \(B'(\sigma)\) are pairwise disjoint for all
capture cells \(\sigma\in\Sigma\), we have 
$\sum_{\sigma\in\Sigma}
|B'(\sigma)|
=
O(n')$.
It therefore suffices to bound
\(\sum_{\sigma\in\Sigma}|\mathcal A'(\sigma)|\).

For every capture cell \(\sigma\), regardless of whether it is Type-I
or Type-II, its parent \(\sigma'\) does not satisfy the capture
condition. Hence,
\(
|\mathcal A'(\sigma')|
<
|\hat A'_{\sigma'}|.
\)

Every point of \(\mathcal A'(\sigma)\) belongs either to
\(\mathcal A'(\sigma')\) or to \(\hat A'_{\sigma'}\) (but not both). Therefore
\[
|\mathcal A'(\sigma)|
\le
|\mathcal A'(\sigma')|
+
|\hat A'_{\sigma'}|
<
2|\hat A'_{\sigma'}|.
\]

Summing over all capture cells and using the fact that every cell has
\(O(1)\) children gives
\[
\sum_{\sigma\in\Sigma}
|\mathcal A'(\sigma)|
=
O\!\left(
\sum_{\sigma'\in\Xi}
|\hat A'_{\sigma'}|
\right).
\]

Since the total size of the conflict lists of all cells of $\Xi$ is \(O(m'r')\), we have
$\sum_{\sigma'\in\Xi}
|\hat A'_{\sigma'}|
=
O(m'r')$.

The lemma follows.
\end{proof}

\begin{lemma}
\label{lem:capture-construct}
All bicliques of \(\mathcal C\) can be computed in
\(O(m'r'+n')\) time.
\end{lemma}

\begin{proof}
We first identify all capture cells. After the hierarchical cutting is
constructed, the conflict list \(\mathcal D'_\sigma\) is available for
every cell \(\sigma\in\Xi\), and therefore
\(|\hat A'_\sigma|\) is also known. We traverse the cutting tree
\(\Xi\) in a top-down manner while maintaining
\(|\mathcal A'(\sigma)|\) for every cell \(\sigma\).
Since each cell has \(O(1)\) children, all values
\(|\mathcal A'(\sigma)|\) can be computed in \(O(m'r')\) time.
Consequently, all capture cells can be identified in \(O(m'r')\) time.

Next, we compute \(\mathcal A'(\sigma)\) for every capture cell
\(\sigma\). For each cell \(\sigma\in\Xi\), let \(A'_\sigma\) denote
the set of points \(a\in A'\) such that the annulus \(D_a\) contains
\(\sigma\) but not the parent of \(\sigma\). As discussed in the proof
of Lemma~\ref{lem:partial-brs}, all sets \(A'_\sigma\) can be computed
in \(O(m'r')\) time.

Observe that \(\mathcal A'(\sigma)\) is precisely the union of
\(A'_{\sigma'}\) over all ancestors \(\sigma'\) of \(\sigma\) (where a
cell is considered an ancestor of itself). Thus, for every capture
cell \(\sigma\), we compute \(\mathcal A'(\sigma)\) by taking the union
of the nonempty sets \(A'_{\sigma'}\) over all ancestors
\(\sigma'\) of \(\sigma\). By Lemma~\ref{lem:bicliquesize}, the total
size of all such sets is \(O(m'r')\), so all sets
\(\mathcal A'(\sigma)\) can be computed within the same bound.

It remains to compute \(B'(\sigma)\) for all capture cells
\(\sigma\). We first locate every point \(b\in B'\) in the last
cutting. As in Section~\ref{subsec:quasi-large}, we use the
\(\Gamma\)-framework and perform each point location in amortized
\(O(1)\) time. Hence, all point locations take \(O(n')\) amortized
time.

If a capture cell \(\sigma\) is a leaf cell, then \(B'(\sigma)\) is
already available. Otherwise, \(B'(\sigma)\) is the union of
\(B'(\sigma'')\) over all leaf descendants \(\sigma''\) of
\(\sigma\). These descendants can be found by traversing the cutting
tree. Therefore, all sets \(B'(\sigma)\) can be computed in
\(O(n')\) time (or \(O(n'+m'r')\) if the tree traversal is included).

Combining the above steps, all bicliques of \(\mathcal C\) can be
computed in \(O(m'r'+n')\) time.
\end{proof}

Combining Lemmas~\ref{lem:bicliquesize} and
\ref{lem:capture-construct}, the subproblem \(T(m',n')\) can be solved
in \(O(m'r'+n')\) time, and the total size of all bicliques produced is
also \(O(m'r'+n')\).

Since \(r'=m'/g^6\), \(m'\le q^{1/3}g^4\), and \(n'\le q^{2/3}g^2\), we obtain \(T(m',n')=O(q^{2/3}g^2)\). Substituting this bound into \eqref{equ:tiny10} yields $T(q,q)=O(q^{4/3}/g^2)$.
Substituting again into \eqref{equ:tiny00} gives $T(n,n)=O(n^{4/3}/g^2)$.
Replacing \(g\) by \(\log^{(i-2)}q\) with $q=\log^6\log n$, we conclude that the total
running time is \(O(n^{4/3}/(\log^{(i)}n)^2)\). By the same analysis,
the total size of all bicliques in \(\Gamma\) is also
\(O(n^{4/3}/(\log^{(i)}n)^2)\).

The running-time and representation-size bounds of
Theorem~\ref{thm:distance-representation} are now established. It
remains to prove the constant-fraction guarantee, namely that the
bicliques of \(\Gamma\) record at least forty percent of the distances
in \(S\cap I\). We prove this in the next subsection.

\subsubsection{Proving the constant-fraction guarantee}

It remains to prove the second statement of
Theorem~\ref{thm:distance-representation}, namely that
\begin{equation}\label{equ:40bound}
\sum_{A_t\times B_t\in\Gamma}
|A_t||B_t|
\ge
0.4\,|S\cap I|.    
\end{equation}

We begin with a single subproblem \(T(m',n')\) discussed above.
Let \(S'\) denote the set of all point pairs between \(A'\) and \(B'\)
whose distances lie in \(I\).

For each leaf cell \(\sigma\) of $\Xi$, the pairs recorded in the biclique
\(\hat A'_\sigma\times B'(\sigma)\) are considered
\emph{uncertain pairs}. All other pairs of \(S'\) are called
\emph{certified pairs}. Since \(|\hat A'_\sigma|\le m'/r'\) and $\sum_{\sigma\in \Xi_k}B'(\sigma)=n'$, the total
number of uncertain pairs in the subproblem is \(O(m'n'/r')\).

The following lemma is the key ingredient. Its proof relies crucially
on the capture condition.

\begin{lemma}
\label{lem:half}
The number of pairs recorded in the bicliques of \(\mathcal C\) is at
least one half of all certified pairs of \(S'\).
\end{lemma}

\begin{proof}
Consider a certified pair \((a,b)\in S'\). Let \(\sigma_b\) be the leaf
cell of $\Xi$ containing \(b\). Since \((a,b)\) is certified, the annulus
\(D_a\) contains \(\sigma_b\).

Assume that \((a,b)\) is not recorded by any biclique of
\(\mathcal C\). Then \(\sigma_b\) cannot itself be a capture cell.
By Lemma~\ref{lem:unique-type1}, there is a unique capture cell
\(\sigma\) on the root-to-leaf path from the root to \(\sigma_b\).
Since \(\sigma_b\) is not a capture cell, \(\sigma\) must be a Type-I
capture cell. Because \(b\in \sigma_b\subseteq \sigma\), we have
\(b\in B'(\sigma)\).

Because \((a,b)\) is not recorded by the biclique
\(\mathcal A'(\sigma)\times B'(\sigma)\), the annulus \(D_a\) does not
contain \(\sigma\). Since \(D_a\) contains
\(\sigma_b\subseteq \sigma\), the boundary of \(D_a\) must cross
\(\sigma\). Hence \(a\in \hat A'_\sigma\).

We charge \((a,b)\) to the Type-I capture cell \(\sigma\).

The number of charged pairs associated with \(\sigma\) is at most
\(|\hat A'_\sigma|\cdot |B'(\sigma)|\).
On the other hand, the number of pairs recorded by the biclique
\(\mathcal A'(\sigma)\times B'(\sigma)\) is
\(|\mathcal A'(\sigma)|\cdot |B'(\sigma)|\).
Since \(\sigma\) is a Type-I capture cell,
\(|\mathcal A'(\sigma)|\ge |\hat A'_\sigma|\). Therefore,
\(
|\mathcal A'(\sigma)|\cdot |B'(\sigma)|
\ge
|\hat A'_\sigma|\cdot |B'(\sigma)|.
\)
Thus every charged pair associated with \(\sigma\) can be injectively
assigned to a pair recorded by
\(\mathcal A'(\sigma)\times B'(\sigma)\).

Summing over all Type-I capture cells, the total number of recorded
pairs in \(\mathcal C\) is at least the total number of charged pairs.
Consequently, the recorded pairs in \(\mathcal C\) account for at least
one half of all certified pairs of \(S'\).
\end{proof}

We now prove \eqref{equ:40bound}.

Let \(N\) denote the total number of uncertain pairs produced by the
algorithm. As shown above, the number of uncertain pairs for a
subproblem \(T(m',n')\) is \(O(m'n'/r')\), which is
\(O(q^{2/3}g^8)\) because \(r'=m'/g^6\) and
\(n'\le q^{2/3}g^2\). It follows that the number of uncertain pairs
associated with a subproblem \(T(q,q)\) in
Recurrence~\eqref{equ:tiny10} is \(O(q^{4/3}g^4)\). Since
\(g=\log^{(i-2)}q\) and \(q=\log^6\log n\), the total number of
uncertain pairs in the entire algorithm is
$N=O(
n^{4/3}(\log^{(i)}n)^4
)$. Hence, for some constant \(c>0\),
\(
N
\le
c\cdot n^{4/3}(\log^{(i)}n)^4.
\)
By the assumption of
Theorem~\ref{thm:distance-representation},
\(
|S\cap I|
\ge
c_1\cdot n^{4/3}(\log^{(i)}n)^4.
\)
Therefore,
\[
N
\le
\frac{c}{c_1}\,|S\cap I|.
\]

Let \(M\) denote the number of pairs recorded in \(\Gamma\), i.e.,
$M
=
\sum_{A_t\times B_t\in\Gamma}
|A_t||B_t|$.

Partition \(M\) into two parts:
\begin{itemize}
\item \(M_1\): the number of pairs recorded by bicliques produced
before the subproblems \(T(m',n')\);

\item \(M_2\): the number of pairs recorded by the capture-cell
bicliques produced inside the subproblems \(T(m',n')\).
\end{itemize}

Hence \(M=M_1+M_2\).

Every pair of \(S\cap I\) belongs to exactly one of the following
categories:

\begin{enumerate}
\item it is already recorded by the bicliques contributing to \(M_1\);
\item it belongs to a certified pair inside some subproblem
\(T(m',n')\);
\item it is an uncertain pair.
\end{enumerate}

By Lemma~\ref{lem:half}, the bicliques contributing to \(M_2\) record
at least one half of all certified pairs inside the subproblems.
Therefore,
\[
M_2
\ge
\frac{|S\cap I|-M_1-N}{2}.
\]

Using \(M=M_1+M_2\), we obtain
\[
M
=
M_1+M_2
\ge
M_1+\frac{|S\cap I|-M_1-N}{2}
=
\frac{M_1+|S\cap I|-N}{2}
\ge
\frac{|S\cap I|-N}{2}.
\]

Since \(N\le (c/c_1)|S\cap I|\),
\[
M
\ge
\frac12
\left(
1-\frac{c}{c_1}
\right)
|S\cap I|.
\]

Choosing \(c_1\ge 5c\) yields
\(
M
\ge
0.4\,|S\cap I|.
\)
This proves \eqref{equ:40bound}. Together with the running-time
and representation-size bounds established earlier, this completes the
proof of Theorem~\ref{thm:distance-representation}.

\section{Distance selection}
\label{sec:selection}

In this section, we combine the quasi-decision theorem
(Theorem~\ref{thm:quasidecision}) and the distance representation theorem
(Theorem~\ref{thm:distance-representation}) to obtain a deterministic
\(O(n^{4/3})\)-time algorithm for distance selection.

Recall that
\(S=\{\|ab\|:a\in A,\ b\in B\}\).
Given an integer \(K\), our goal is to compute the \(K\)-th smallest
distance in \(S\), denoted by \(\delta^*\).

The algorithm proceeds in phases. Throughout the execution, we maintain
an interval \(I\) containing \(\delta^*\). 
For \(i=1,2,\ldots,O(\log^*n)\), the goal of the \(i\)-th phase is to
transform an interval \(I\) satisfying

\[
|S\cap I|
=
O\!\left(
n^{4/3}(\log^{(i-1)}n)^4
\right)
\]
(and \(O(n^2)\) when \(i=1\)) and containing \(\delta^*\) into a smaller interval \(I'\subseteq I\)
such that

\[
|S\cap I'|
=
O\!\left(
n^{4/3}(\log^{(i)}n)^4
\right)
\]
and \(\delta^*\in I'\).

Initially, we have \(I=(0,+\infty)\), and hence
\(|S\cap I|=O(n^2)\). After \(O(\log^*n)\) phases, we obtain an interval
\(I\) satisfying \(|S\cap I|=O(n^{4/3})\) and containing \(\delta^*\).
A final {\em distance-reporting} step then computes \(\delta^*\) from \(S\cap I\)
in additional \(O(n^{4/3})\) time.

We will show that the \(i\)-th phase consists of
\(O(\log^{(i)}n)\) iterations. Each iteration, which reduces the remaining candidate distances by a constant fraction, takes
\(O(n^{4/3}/\log^{(i)}n)\) time for \(i=1,2\), and
\(O(n^{4/3}/(\log^{(i)}n)^2)\) time for \(i\ge 3\).
Consequently, each of the first two phases runs in
\(O(n^{4/3})\) time, while the \(i\)-th phase for \(i\ge 3\) runs in
\(O(n^{4/3}/\log^{(i)}n)\) time.
Therefore, the total time spent over all phases is
\[
O(n^{4/3})\cdot 
\left(
1+
\sum_{i=3}^{O(\log^*n)}
\frac{1}{\log^{(i)}n}
\right),
\]
which is \(O(n^{4/3})\). Including the final distance-reporting step, the
overall running time remains \(O(n^{4/3})\).

In what follows, we first describe the algorithm in
Section~\ref{subsec:algodes} and then prove its correctness in
Section~\ref{subsec:algocor}.

\subsection{Algorithm description}
\label{subsec:algodes}

We describe a general \(i\)-th phase. Before the phase starts, we are
given an interval \(I\) satisfying $|S\cap I| = O(n^{4/3}(\log^{(i-1)}n)^4)$ and containing \(\delta^*\). The goal of the phase is to produce a subinterval \(I'\subseteq I\) such that $|S\cap I'| = O(n^{4/3}(\log^{(i)}n)^4)$ and \(\delta^*\in I'\).

The phase proceeds in iterations. At the beginning of the \(j\)-th
iteration, we are given an interval \(I_{j-1}\) containing
\(\delta^*\). Initially, \(I_0=I\). Each iteration either:
\begin{enumerate}
\item produces a smaller interval \(I_j\subseteq I_{j-1}\) satisfying
      \(\delta^*\in I_j\) and
      \(|S\cap I_j|\le \rho\,|S\cap I_{j-1}|\) for some constant
      \(\rho<1\); or

\item certifies that
      \(|S\cap I_{j-1}|=O(n^{4/3}(\log^{(i)}n)^4)\), in which case
      the phase terminates and returns \(I_{j-1}\) as $I'$.
\end{enumerate}

Since the phase starts with
\(|S\cap I|=O(n^{4/3}(\log^{(i-1)}n)^4)\), the first outcome can occur
only \(O(\log^{(i)}n)\) times. Hence, the phase terminates after
\(O(\log^{(i)}n)\) iterations.

In the following, we describe a general $j$-th iteration. For simplicity, we
assume \(i\ge 3\). The cases \(i=1,2\) are identical except for the
running-time bounds. We will show that each iteration requires
\(O(n^{4/3}/(\log^{(i)}n)^2)\) time. 

We begin by applying
Theorem~\ref{thm:distance-representation} to the interval \(I_{j-1}\)
with parameter \(i\). This produces a collection of bicliques
$\Gamma=\{A_t\times B_t\}$ in \(O(n^{4/3}/(\log^{(i)}n)^2)\) time.
Let
\(N(\Gamma) = \sum_{A_t\times B_t\in\Gamma} |A_t||B_t|\), i.e., the number of distances recorded in $\Gamma$. 
Since
\(\sum_t(|A_t|+|B_t|)
=
O(n^{4/3}/(\log^{(i)}n)^2)\)
by Theorem~\ref{thm:distance-representation}, the value
\(N(\Gamma)\) can be computed within the same asymptotic running time.

In the remainder of the section, let \(c_1\) denote the constant from
Theorem~\ref{thm:distance-representation}.

\begin{lemma}\label{lem:bounded}
If $N(\Gamma) \le c_1\,n^{4/3}(\log^{(i)}n)^4$, 
then $|S\cap I_{j-1}| = O(n^{4/3}(\log^{(i)}n)^4)$.
\end{lemma}
\begin{proof}
If
\(|S\cap I_{j-1}|<
c_1n^{4/3}(\log^{(i)}n)^4\),
the conclusion is immediate.
Otherwise, Theorem~\ref{thm:distance-representation} implies that
$N(\Gamma) \ge 0.4\,|S\cap I_{j-1}|$.
Therefore,
\[
|S\cap I_{j-1}|
\le
\frac{1}{0.4}N(\Gamma)
\le
\frac{c_1}{0.4}
n^{4/3}(\log^{(i)}n)^4,
\]
which proves the lemma.
\end{proof}

If
\(N(\Gamma)\le c_1n^{4/3}(\log^{(i)}n)^4\),
Lemma~\ref{lem:bounded} allows us to terminate the current phase and return
\(I_{j-1}\) as $I'$.

Assume henceforth that
$N(\Gamma)
>
c_1n^{4/3}(\log^{(i)}n)^4$.
Since $|S\cap I_{j-1}|\geq N(\Gamma)$, we have $|S\cap I_{j-1}| > c_1n^{4/3}(\log^{(i)}n)^4$. Consequently, by Theorem~\ref{thm:distance-representation}, 
$N(\Gamma)
\ge
0.4\,|S\cap I_{j-1}|$.

We next construct expander graphs for the bicliques of \(\Gamma\) and
use them to compute a set \(E\) of weighted distances. This step is
essentially the same as in \cite{ref:WangIm25}, which in turn follows the method in \cite{ref:KatzAn97}. For
completeness, we briefly sketch the construction and refer the reader
to \cite[Section~3]{ref:WangIm25} or \cite{ref:KatzAn97} for details.

\paragraph{Computing a weighted distance set $\boldsymbol{E}$.}
For convenience, assume that \(|A_t|\ge |B_t|\) for every biclique \(A_t\times B_t\in\Gamma\) (otherwise we could switch the names of $A_t$ and $B_t$). Let
\(m_t=|A_t|\) and \(n_t=|B_t|\).

Partition each \(A_t\) into
\(g=\lfloor m_t/n_t\rfloor\) subsets
\(A_{t1},A_{t2},\ldots,A_{tg}\), each containing between
\(n_t\) and \(2n_t\) points. Each pair
\((A_{tf},B_t)\), \(1\le f\le g\), defines a biclique
\(A_{tf}\times B_t\).
For every such biclique, we construct a \(d\)-regular
LPS-expander graph \(G_{tf}\) on
\(A_{tf}\cup B_t\)~\cite{ref:KatzAn97,ref:LubotzkyEx86},
where \(d\) is a sufficiently large constant to be fixed later.
The graph \(G_{tf}\) has \(O(|A_{tf}|+|B_t|)\) edges and can be
constructed within the same time bound~\cite{ref:KatzAn97,ref:LubotzkyEx86}.
Let \(G_t\) be the union of all graphs \(G_{tf}\).
Computing all graphs \(G_t\) takes
\(O(\sum_t(|A_t|+|B_t|))
=
O(n^{4/3}/(\log^{(i)}n)^2)\) time.
The total number of edges in all graphs is bounded by the same
bound.

For every edge \((a,b)\) of \(G_t\), where
\(a\in A_t\) and \(b\in B_t\), we associate the distance
\(\|ab\|\). The resulting collection of distances is denoted by \(E\).
The size of \(E\) is therefore
\(O(n^{4/3}/(\log^{(i)}n)^2)\). Moreover, every distance of \(E\)
lies in the interval \(I_{j-1}\).

We next assign a weight to every value of \(E\).
The weight assignment is chosen so that the total weight of all
distances of $E$ is proportional to \(N(\Gamma)\).
Since every graph \(G_{tf}\) is \(d\)-regular, it has at most
\((|A_{tf}|+|B_t|)d/2\) edges.
We assign each distance contributed by \(G_{tf}\) the weight
\[
\frac{|A_{tf}||B_t|}
     {|A_{tf}|+|B_t|}.
\]

The total weight of all values in \(E\) is therefore at most
   \[
        \sum_{t, f} (|A_{tf}| + |B_t|) \cdot \frac{d}{2} \cdot \frac{|A_{tf}| |B_t|}{|A_{tf}| + |B_t|} = \frac{d}{2}\cdot  \sum_{t, f} |A_{tf}| |B_t| =\frac{d}{2}\cdot N(\Gamma).
    \]

Recall that 
\(N(\Gamma)>
c_1n^{4/3}(\log^{(i)}n)^4\).
We may assume \(N(\Gamma)\ge 16\).
Furthermore, every value of \(E\) has weight at most
\(|A_{tf}||B_t|/(|A_{tf}|+|B_t|)\le
\frac12\sqrt{|A_{tf}||B_t|}
\le \frac12\sqrt{N(\Gamma)}
\le N(\Gamma)/8\).

Using these weights, we choose values
\(\delta_1<\delta_2<\cdots<\delta_h\) from \(E\),
where \(1\le h\le 2d\), so that they partition
\(I_{j-1}\) into \(h+1\) intervals and the total weight of the values
of \(E\) in each interval lies between
\(N(\Gamma)/4\) and \(N(\Gamma)/2\).
The partition can be computed in \(O(|E|)\) time using the linear-time
selection algorithm~\cite{ref:BlumTi73}, and therefore requires
\(O(n^{4/3}/(\log^{(i)}n)^2)\) time.

\paragraph{Computing \texorpdfstring{$I_j$}{Ij} by testing the values of \texorpdfstring{$E$}{E}.}

We now use the values of \(E\) to compute the next interval \(I_j\).

Assume that \(I_{j-1}=(\alpha,\beta]\). Let
\(\delta_0=\alpha\) and \(\delta_{h+1}=\beta\). For
\(1\le u\le h+1\), define
\(I'_u=(\delta_{u-1},\delta_u]\). Let
\(\mathcal I=\{I'_u\mid 1\le u\le h+1\}\). Clearly, the intervals of
\(\mathcal I\) form a partition of \(I_{j-1}\).

Ideally, we would like to identify the interval of \(\mathcal I\)
containing \(\delta^*\) and return it as \(I_j\).
Previous deterministic algorithms accomplish this by repeatedly
invoking an exact decision procedure, which computes, for a given value
\(\delta\), the number \(K_\delta\) of distances in \(S\) that are at
most \(\delta\). Since \(\delta<\delta^*\) if and only if
\(K_\delta<K\), such a procedure can be used to locate the interval
containing \(\delta^*\).

The exact decision algorithm requires \(O(n^{4/3})\) time~\cite{ref:WangUn23,ref:ChanHo23}. Consequently, using it in every
iteration leads to the \(O(n^{4/3}\log n)\)-time algorithm of
Wang and Zhao~\cite{ref:WangIm25}. To achieve an
\(O(n^{4/3})\)-time algorithm, we instead use the quasi-decision
algorithm of Theorem~\ref{thm:quasidecision}.

For every \(u\), \(1\le u\le h\), we apply
Theorem~\ref{thm:quasidecision} to the interval
\(I''_u=[0,\delta_u]\). Note that $\kappa(I''_u)=K_{\delta_u}$.

Let \(u_1\) be the largest index such that the algorithm of Theorem~\ref{thm:quasidecision} reports
\(\kappa(I''_{u_1})<K\) (implying that $\delta_{u_1}<\delta^*$). If no such index exists, we set \(u_1=0\).
Since \(\delta_0=\alpha\) and
\(\delta^*\in I_{j-1}=(\alpha,\beta]\), we always have
\(\delta_{u_1}<\delta^*\).

Similarly, let \(u_2\) be the smallest index such that the algorithm
reports \(\kappa(I''_{u_2})>K\) (implying that $\delta^*\leq \delta_{u_2}$). If no such index exists, we set
\(u_2=h+1\). In either case, \(\delta^*\le \delta_{u_2}\).

Therefore, $\delta^* \in (\delta_{u_1},\delta_{u_2}]$. 

We terminate the current iteration and return $I_j=(\delta_{u_1},\delta_{u_2}]$.

Since \(h\le 2d\) and \(d=O(1)\), only \(O(1)\) calls to
Theorem~\ref{thm:quasidecision} are required. Each call takes
\(O(n^{4/3}/(\log^{(i)}n)^2)\) time. Therefore, the total running time
of the iteration remains
\(O(n^{4/3}/(\log^{(i)}n)^2)\).

The correctness of the algorithm is established by the following
lemma, whose proof is given in the next subsection.

\begin{lemma}
\label{lem:prune}
There exists a constant \(\rho\in(0,1)\) such that
$|S\cap I_j|
\le
\rho\,|S\cap I_{j-1}|$.
\end{lemma}

\subsection{Algorithm correctness: Proving Lemma~\ref{lem:prune}}
\label{subsec:algocor}


Let \(n_j=|S\cap I_j|\) and \(n_{j-1}=|S\cap I_{j-1}|\). Our goal is
to show that \(n_j<\rho\cdot n_{j-1}\) for some constant
\(\rho\in(0,1)\).

For notational convenience, let \(m_{j-1}=N(\Gamma)\), i.e.,
\(m_{j-1}\) is the number of distances of \(S\cap I_{j-1}\) recorded
in \(\Gamma\). Define \(m'_{j-1}\) as the number of distances of
\(S\cap I_{j-1}\) not recorded in \(\Gamma\). Then
\(n_{j-1}=m_{j-1}+m'_{j-1}\).

By Theorem~\ref{thm:distance-representation},
\(m_{j-1}\ge 0.4\,n_{j-1}\), and therefore
\(m'_{j-1}\le 0.6\,n_{j-1}\).

Let \(S_j\) denote the set of distances of \(S\cap I_j\) not recorded
in \(\Gamma\), and \(S_{j-1}\) the set of distances of \(S\cap I_{j-1}\) not recorded
in \(\Gamma\). Since
\(I_j\subseteq I_{j-1}\), we have \(S_j\subseteq S_{j-1}\), and hence
\(|S_j|\le |S_{j-1}|=m'_{j-1}\).

Define \(q_j\) as the number of distances of \(S\cap I_j\) recorded in
\(\Gamma\). Then
\[
n_j=q_j+|S_j|
\le q_j+m'_{j-1}
\le q_j+0.6\,n_{j-1}.
\]

Since \(m_{j-1}\le n_{j-1}\), Lemma~\ref{lem:prune} follows
immediately from the following lemma.

\begin{lemma}
\label{lem:prune10}
There exists a constant \(\rho'\in(0,0.4)\) such that
\(q_j\le \rho' \cdot m_{j-1}\).
\end{lemma}

\paragraph{Proof of Lemma~\ref{lem:prune10}.}

We refer to the intervals of \(\mathcal I\) as \emph{basic intervals}.
Recall that \(\delta^*\in I_j\), which implies
\(\delta_{u_1}<\delta_{u_2}\) and therefore \(u_1<u_2\).

If \(u_1=u_2-1\), then \(I_j\) consists of the single basic interval
\(I'_{u_1+1}\). If \(u_1=u_2-2\), then \(I_j\) consists of the two
basic intervals \(I'_{u_1+1}\) and \(I'_{u_2}\). Otherwise,
\(u_1<u_2-2\), and \(I_j\) can be partitioned into three intervals:
\(I'_{u_1+1}\), \((\delta_{u_1+1},\delta_{u_2-1}]\), and \(I'_{u_2}\).
In the following, we only analyze the third case (i.e., $u_1<u_2-2$) since it is the most general case and the other two cases can be analyzed similarly but in a much easier way. 

Let \(\hat I=(\delta_{u_1+1},\delta_{u_2-1}]\), and define
\(\Delta=c_0\cdot n^{4/3}(\log^{(i)}n)^4\), which is the term in Theorem~\ref{thm:quasidecision}(3).

\begin{lemma}[Middle interval bound]
\label{lem:middleinterval}
\( |S\cap \hat I|\le 2\Delta. \)
\end{lemma}

\begin{proof}
Since \(u_1<u_1+1<u_2-1<u_2\), the definition of \(u_1\) implies that
the quasi-decision algorithm of Theorem~\ref{thm:quasidecision} reports ``unsure'' on
\(I''_{u_1+1}=[0,\delta_{u_1+1}]\). Therefore,
\(|\kappa(I''_{u_1+1})-K|\le \Delta\).

Similarly, the algorithm reports ``unsure'' on
\(I''_{u_2-1}=[0,\delta_{u_2-1}]\), and thus
\(|\kappa(I''_{u_2-1})-K|\le \Delta\).

Consequently,
\[
|S\cap \hat I|
=
\kappa(I''_{u_2-1})
-
\kappa(I''_{u_1+1})
\le
2\Delta.
\]

The lemma thus follows. 
\end{proof}

Let \(I^*\) be the union of the two boundary basic intervals
\(I'_{u_1+1}\) and \(I'_{u_2}\). Let \(S(\Gamma)\) denote the set of
distances recorded in \(\Gamma\). Since \(\hat I\) and \(I^*\) form a
partition of \(I_j\),
\begin{equation}
\label{equ:prune00}
q_j
=
|S(\Gamma)\cap \hat I|
+
|S(\Gamma)\cap I'|.
\end{equation}

The first term is controlled by Lemma~\ref{lem:middleinterval}, while
the second term will be bounded using the expander analysis.

By Lemma~\ref{lem:middleinterval},
\[
|S(\Gamma)\cap \hat I|
\le
|S\cap \hat I|
\le
2\Delta= 2c_0\cdot n^{4/3}(\log^{(i)}n)^4.
\]

Since \(m_{j-1}>c_1 n^{4/3}(\log^{(i)}n)^4\), we obtain

\begin{equation}
\label{equ:prune10}
|S(\Gamma)\cap \hat I|
<
\frac{2c_0}{c_1}\,m_{j-1}.
\end{equation}

To bound the second term of~\eqref{equ:prune00}, we use the following
lemma.

\begin{lemma}
\label{lem:prune20}
There exists a constant \(\rho''\in(0,0.18)\) such that
\(|S(\Gamma)\cap I''|<\rho''\,m_{j-1}\) for every basic interval
\(I''\) of \(I^*\).
\end{lemma}
\begin{proof}
The proof follows the previous work~\cite{ref:WangIm25,ref:KatzAn97}.

Consider a graph \(G_{tf}\) computed from \((A_{tf}, B_t)\) in our algorithm, where \(1 \leq f \leq g\). Some edges of \(G_{tf}\) contribute distances in \(E\). We partition all such graphs \(G_{tf}\) into two sets. Let \(\calG_1\) denote the set of graphs \(G_{tf}\) that contribute fewer than \(3(|A_{tf}| + |B_t|)\) distances in \(E \cap I''\), and let \(\calG_2\) denote the set of the remaining graphs.

\paragraph{The set \(\boldsymbol{\calG_1}\).}
Let \(I''=(\delta',\delta'']\). For a graph \(G_{tf} \in \calG_1\), built on the pair \((A_{tf}, B_t)\), let \(\calD_{tf}\) be the set of annuli centered at the points of \(A_{tf}\) with radii \(\delta'\) and \(\delta''\). Consider a \(1/r\)-cutting \(\Xi\) for \(\calD_{tf}\), where \(r\) is a constant to be chosen later. The cutting partitions the plane into \(O(r^2)\) cells so that, for every cell \(\sigma\), at most \(|\calD_{tf}|/r\) annuli in \(\calD_{tf}\) have a boundary circle intersecting the interior of \(\sigma\).

For each cell \(\sigma \in \Xi\), let \(B_t(\sigma)\) denote the points of \(B_t\) inside \(\sigma\), let \(A_{tf}(\sigma)\) be the points of \(A_{tf}\) whose annuli fully contain \(\sigma\), and let \(A_{tf}'(\sigma)\) be the points of \(A_{tf}\) whose annuli have at least one boundary circle intersecting \(\sigma\). Let \(N_{tf}\) denote the number of interpoint distances between \(A_{tf}\) and \(B_t\) that lie in \(I''\). Then
\[
N_{tf} \leq \sum_{\sigma \in \Xi} |A_{tf}(\sigma)|\,|B_t(\sigma)|
       + \sum_{\sigma \in \Xi} |A_{tf}'(\sigma)|\,|B_t(\sigma)|.
\]
Since \(|A_{tf}'(\sigma)| \leq |A_{tf}|/r\) and \(\sum_{\sigma \in \Xi} |B_t(\sigma)| = |B_t|\), we obtain
\[
\sum_{\sigma \in \Xi} |A_{tf}'(\sigma)|\,|B_t(\sigma)|
= O\!\left(\frac{|A_{tf}|\,|B_t|}{r}\right).
\]

Now consider \(\sum_{\sigma \in \Xi} |A_{tf}(\sigma)|\,|B_t(\sigma)|\). Note that for any \(a \in A_{tf}(\sigma)\) and \(b \in B_t(\sigma)\), we have \(\|ab\| \in I''\). If the edge \((a,b)\) exists in \(G_{tf}\), then \(\|ab\|\) must lie in \(E \cap I''\). Since \(G_{tf} \in \calG_1\), it has fewer than \(3(|A_{tf}| + |B_t|)\) such edges. By \cite[Corollary 2.5]{ref:KatzAn97}, if \(X\) and \(Y\) are two vertex subsets of a \(d\)-regular expander graph on \(M\) vertices and there are fewer than \(3M\) edges between \(X\) and \(Y\), then \(|X|\,|Y| \leq 9M^2/d\). Applying this with \(X=A_{tf}(\sigma)\), \(Y=B_t(\sigma)\), and \(M=|A_{tf}|+|B_t|\), we get
\[
\sum_{\sigma \in \Xi} |A_{tf}(\sigma)|\,|B_t(\sigma)|
\leq O(r^2)\cdot \frac{9(|A_{tf}|+|B_t|)^2}{d}
= O\!\left(\frac{r^2(|A_{tf}|+|B_t|)^2}{d}\right).
\]

Therefore,
\[
N_{tf}
= O\!\left(\frac{|A_{tf}|\,|B_t|}{r}\right)
  + O\!\left(\frac{r^2(|A_{tf}|+|B_t|)^2}{d}\right).
\]
Since \(|B_t| \leq |A_{tf}| \leq 2|B_t|\), we have \((|A_{tf}|+|B_t|)^2 \leq 5|A_{tf}|\,|B_t|\). Hence,
\[
N_{tf}
= O\!\left(\left[\frac1r + \frac{r^2}{d}\right]|A_{tf}|\,|B_t|\right).
\]
Choosing \(r=d^{1/3}\) and taking \(d\) sufficiently large, we obtain
\[
N_{tf} \leq 0.01\cdot |A_{tf}|\,|B_t|.
\]

Summing over all graphs \(G_{tf}\in\calG_1\) gives
\[
N(\calG_1) \leq 0.01 \cdot \sum_{G_{tf}\in\calG_1} |A_{tf}|\,|B_t|,
\]
 where $N(\calG_1)$ is the number of distances between points of $A_{tf}$ and points of $B_t$ that are in $I''$ for all graphs $G_{tf} \in \calG_1$.
 
Since
\[
\sum_{G_{tf}\in\calG_1} |A_{tf}|\,|B_t|
\leq \sum_t |A_t|\,|B_t|
= m_{j-1},
\]
we conclude that \(N(\calG_1)\leq 0.01\cdot m_{j-1}\).

\paragraph{The set \(\boldsymbol{\calG_2}\).}
For each \(G_{tf}\in\calG_2\), at least \(3(|A_{tf}|+|B_t|)\) distances in \(E\cap I''\) are contributed, and therefore \(G_{tf}\) contributes at least \(3|A_{tf}|\,|B_t|\) to the total weight of distances in \(E\cap I''\). Recall that the total weight of distances in \(E\cap I''\) is at most \(m_{j-1}/2\) by our algorithm. Hence,
\[
\sum_{G_{tf}\in\calG_2} |A_{tf}|\,|B_t| \leq \frac{m_{j-1}}{6}.
\]
Let \(N(\calG_2)\) denote the number of distances between points of \(A_{tf}\) and points of \(B_t\) that lie in \(I''\) for all graphs \(G_{tf}\in\calG_2\). Then
\[
N(\calG_2) \leq \sum_{G_{tf}\in\calG_2} |A_{tf}|\,|B_t|
\leq \frac{m_{j-1}}{6}.
\]

\paragraph{Summary.}
By definition, \(|S(\Gamma)\cap I''| = N(\calG_1) + N(\calG_2)\). Since
\(N(\calG_1)\leq 0.01\,m_{j-1}\) and \(N(\calG_2)\leq m_{j-1}/6\), we obtain
\[
|S(\Gamma)\cap I''| \leq \left(0.01 + \frac16\right)m_{j-1}.
\]
Setting \(\rho'' = 0.01 + 1/6\), we have \(\rho'' \in (0,0.18)\), and the lemma follows.
\end{proof}

By Lemma~\ref{lem:prune20}, and since \(I^*\) consists of two basic intervals, we have
\[
|S(\Gamma)\cap I'|
\leq 2\rho''\, m_{j-1}
< 0.36\, m_{j-1},
\]
where \(\rho''\in(0,0.18)\).

Combining this bound with \eqref{equ:prune00} and \eqref{equ:prune10}, we obtain
\[
q_j
\leq \frac{2c_0}{c_1}\,m_{j-1}
   + 0.36\,m_{j-1}.
\]

Choosing \(c_1\) sufficiently large so that \(2c_0/c_1 \leq 0.01\) yields
$q_j \leq 0.37\,m_{j-1}$,
which proves Lemma~\ref{lem:prune10}.

\subsection{The distance-reporting step}
\label{sec:post}

As discussed above, after \(O(\log^* n)\) phases, we are left with an interval \(I=(\alpha,\beta]\) such that \(|S\cap I|=O(n^{4/3})\) and \(\delta^*\in I\). In this distance-reporting step, we compute \(\delta^*\) in an additional \(O(n^{4/3})\) time.

We first compute the number of distances in \(S\) that are at most \(\alpha\), which can be done in \(O(n^{4/3})\) time~\cite{ref:WangUn23}. Let \(K_\alpha\) denote this number. Our goal is then to find the \(K'\)-th smallest distance in \(S\cap I\), where \(K' = K-K_\alpha\).
To this end, we first explicitly enumerate all distances in \(S\cap I\), and then apply the linear-time selection algorithm to compute \(\delta^*\). Since \(|S\cap I|=O(n^{4/3})\), the selection step takes \(O(n^{4/3})\) time. Therefore, the remaining task is to enumerate \(S\cap I\) within the same time bound.

A natural approach is to apply Lemma~\ref{lem:partial-brs} to compute \(\Gamma\) and \(\Pi\) with \(\Pi=\emptyset\). We refer to this as a \emph{full biclique partition} (called \emph{full batched range searching} in \cite{ref:WangIm25}). Since \(\Pi=\emptyset\), the bicliques in \(\Gamma\) record all distances in \(S\cap I\), and thus all such distances can be enumerated from \(\Gamma\) in additional \(O(|S\cap I|)\) time.
To compute a full biclique partition, the algorithm of~\cite{ref:WangIm25} recursively applies Lemma~\ref{lem:partial-brs}. However, the resulting running time is \(n^{4/3}\cdot 2^{O(\log^* n)}\) because of the constant-factor blowup in the recurrence~\cite{ref:WangIm25}. Moreover, the bound obtained in~\cite{ref:WangIm25} for
$\sum_{A_t\times B_t\in \Gamma} (|A_t|+|B_t|)$ is also \(n^{4/3}\cdot 2^{O(\log^* n)}\), again due to the constant-factor blowup. 

\paragraph{Implicit biclique partition.}
To obtain an \(O(n^{4/3})\)-time algorithm for enumerating \(S\cap I\), one would ideally compute a full biclique partition \(\Gamma=\{A_t\times B_t\}\) in \(O(n^{4/3})\) time such that
$\sum_t (|A_t|+|B_t|)=O(n^{4/3})$. We achieve a weaker form of this goal in which the sets \(A_t\) are represented only implicitly.

Specifically, our algorithm outputs the sets \(B_t\) explicitly, together with a collection \(\calA\) of explicitly represented subsets \(A_s'\subseteq A\). For each \(B_t\), the algorithm also outputs a subcollection \(\calA_t\subseteq \calA\) such that the sets in \(\calA_t\) form a partition of \(A_t\). In particular,
$A_t=\cup_{A_s'\in \calA_t} A_s'$.
Furthermore, given \(B_t\), the sets of \(\calA_t\) can be accessed one at a time in \(O(1)\) time per set. Consequently, the points of \(A_t\) can be enumerated in \(O(|A_t|)\) time.
We refer to the pair \((\Gamma,\calA)\) as an \emph{implicit full biclique partition} of \(S\cap I\). Our main result is stated below.

\begin{theorem}[Implicit biclique partition]
\label{thm:implicitbiclique}
Given an interval \(I\), one can compute an implicit full biclique partition of \(S\cap I\) in \(O(n^{4/3})\) time.

The output consists of a collection \(\calA\) of subsets \(A_s'\subseteq A\) and a collection \(\Gamma\) of edge-disjoint bicliques \(A_t\times B_t\) that together record all distances in \(S\cap I\). The sets \(B_t\subseteq B\) are represented explicitly, whereas each set \(A_t\subseteq A\) is represented implicitly as follows. For every biclique \(A_t\times B_t\), there is an associated subcollection \(\calA_t\subseteq \calA\) containing \(O(\log\log\log n)\) subsets \(A_s'\), and these subsets form a partition of \(A_t\). Moreover, the elements of \(\calA_t\) can be accessed one at a time in \(O(1)\) time per element.

Finally,
\[
\sum_{A_t\times B_t\in\Gamma}|B_t| = O(n^{4/3})
\qquad\text{and}\qquad
\sum_{A_s'\in\calA}|A_s'| = O(n^{4/3}).
\]
\end{theorem}
\begin{proof}
We run the algorithm of Lemma~\ref{lem:partial-brs} from Section~\ref{subsec:pbs} with \(m=n\). Whenever the algorithm reaches Recurrence~\eqref{formula:PrimalRecurrence}, for each cell \(\sigma\in\Xi\) we add the biclique \(A_\sigma\times B_\sigma\) to \(\Gamma\), add \(A_\sigma\) to \(\calA\) as a subset \(A_s'\), and define \(\calA_\sigma=\{A_\sigma\}\). Thus \(\calA_\sigma\) contains exactly one subset.

We then continue with the dual algorithm and obtain Recurrence~\eqref{formula:DualRecurrence}. We update \(\Gamma\) and \(\calA\) analogously for the dual subproblems. In this way we obtain Recurrence~\eqref{formula:balanced}. Note that these additional bookkeeping operations do not affect the asymptotic running time, so the recurrence remains unchanged. Setting \(r=n^{1/3}/\log^2 n\) yields Recurrence~\eqref{equ:tiny00}.

To solve \(T(q,q)\), unlike in Section~\ref{subsec:quasi-large}, we do not use the parameter \(g\) (or equivalently, we set \(g=1\)). Applying the algorithm corresponding to Recurrence~\eqref{formula:DualRecurrence} with \(m=n=q\) and \(r=q^{1/3}\), we obtain
\begin{equation}
\label{equ:tiny100}
T(q,q)
=
O(q^{4/3})
+
O(q^{2/3})\cdot T(q^{1/3},q^{2/3}).
\end{equation}

As before, every biclique generated during this step is inserted into \(\Gamma\), and its associated subset is inserted into \(\calA\). At this stage, each biclique \(A_t\times B_t\in\Gamma\) is associated with a single subset \(A_s'\in\calA\).

We next consider a subproblem \(T(m',n')\) arising from Recurrence~\eqref{equ:tiny100}, where \(m'\le q^{1/3}\) and \(n'\le q^{2/3}\). Let \(A'\subseteq A\) and \(B'\subseteq B\) denote the corresponding point sets.

Let \(\calD'\) be the set of annuli defined by the points of \(A'\) with respect to the interval \(I\). We construct a hierarchical \((1/r')\)-cutting
$\{\Xi_0,\Xi_1,\ldots,\Xi_k\}$
for \(\calD'\), with \(r'=m'\). This takes \(O(m'r')\) time. Since \(r'=m'\), we can assume that every cell in the last cutting \(\Xi_k\) has an empty conflict list. Equivalently, no annulus of \(\calD'\) has a boundary circle intersecting the interior of such a cell. Consequently, there are no uncertain pairs, which will allow us to obtain a full biclique partition.

Let \(\Xi\) denote the set of all cells in the hierarchical cutting. For each cell \(\sigma\in\Xi\), define the canonical subset \(A'_\sigma\) exactly as in Lemma~\ref{lem:partial-brs}: a point \(a\in A'\) belongs to \(A'_\sigma\) if the annulus \(D_a\) contains \(\sigma\) but does not contain the parent of \(\sigma\). We insert each nonempty subset \(A'_\sigma\) into \(\calA\). All canonical subsets can be computed in \(O(m'r')\) time.

In addition, for every cell \(\sigma\in\Xi\) with \(A'_\sigma\neq\emptyset\), for a reason to be clear later, we store an \emph{ancestor link} from \(\sigma\) to its lowest ancestor \(\sigma'\) satisfying $\sigma'\neq \sigma$ and \(A'_{\sigma'}\neq\emptyset\). These links can be established in \(O(m'r')\) time by traversing the cutting tree.

Next, for each point \(b\in B'\), we determine the cell of \(\Xi_k\) containing \(b\). Using the \(\Gamma\)-algorithm framework~\cite{ref:ChanHo23}, as discussed in Section~\ref{subsec:quasi-large}, this can be done in \(O(n')\) amortized time for all points $b\in B'$.

For each cell \(\sigma\in\Xi_k\), let \(B_t=B'(\sigma)\), and let \(A_t\) be the union of all canonical subsets \(A'_{\sigma'}\) over the ancestors \(\sigma'\) of \(\sigma\). We add the biclique \(A_t\times B_t\) to \(\Gamma\). We define \(\calA_t\) to be the collection of all nonempty canonical subsets \(A'_{\sigma'}\) corresponding to ancestors \(\sigma'\) of \(\sigma\). By following the ancestor links, the elements of \(\calA_t\) can be accessed one at a time in \(O(1)\) time per element.
Since the height of the cutting tree is \(O(\log r')\), we have
$|\calA_t| = O(\log r')$.
Further, since \(r'=m'\le q^{1/3}\) and \(q=O(\log^6\log n)\), it follows that
$|\calA_t| = O(\log\log\log n)$.

The bicliques generated above for the cells $\sigma$ of \(\Xi_k\) form a full biclique partition for all distances between points of \(A'\) and points of \(B'\). Therefore, the total running time for the subproblem $T(m',n')$ is
$O(m'r' + n')$. Since \(r'=m'\), \(m'\le q^{1/3}\), and \(n'\le q^{2/3}\), we have
\[T(m',n') = O(m'^2+n') = O(q^{2/3}).\]
Substituting this bound into \eqref{equ:tiny100} yields
$T(q,q)=O(q^{4/3})$. Finally, substituting this bound into Recurrence~\eqref{equ:tiny00} gives
$T(n,n)=O(n^{4/3})$, which proves Theorem~\ref{thm:implicitbiclique}.
\end{proof}

\begin{corollary}
\label{cor:outputdistance}
Given an interval \(I\), one can explicitly enumerate all distances in \(S\cap I\) in
$O(n^{4/3}+|S\cap I|)$ time.
\end{corollary}
\begin{proof}
Apply Theorem~\ref{thm:implicitbiclique}. For each biclique \(A_t\times B_t\in\Gamma\), we first enumerate the points of \(A_t\) in \(O(|A_t|)\) time using its implicit representation. Then, for every \(a\in A_t\) and \(b\in B_t\), we output the distance \(\|ab\|\).
By Theorem~\ref{thm:implicitbiclique},
\[
\sum_{A_t\times B_t\in\Gamma}|B_t|=O(n^{4/3})
\qquad\text{and}\qquad
\sum_{A_s'\in\calA}|A_s'|=O(n^{4/3}).
\]
Hence the total overhead for recovering all sets \(A_t\) is \(O(n^{4/3})\), while the time spent outputting distances is proportional to the number of distances reported, namely \(|S\cap I|\). Therefore, the overall running time is \(O(n^{4/3}+|S\cap I|)\).
\end{proof}

By Corollary~\ref{cor:outputdistance},  our original problem of enumerating all distances of $S\cap I$ can be solved in  $O(n^{4/3})$ time, since $|S\cap I|=O(n^{4/3})$.


\paragraph{Wrapping things up.}
We have shown that \(\delta^*\) can be computed in \(O(n^{4/3})\) deterministic time. Note that the algorithm of Corollary~\ref{cor:outputdistance} can also report, together with each distance, the corresponding pair of points that realizes it. Therefore, a point pair whose distance is equal to \(\delta^*\) can be reported within the same time bound. 

The following theorem summarizes our main result. 

\begin{theorem}
Given two sets \(A\) and \(B\), each consisting of \(n\) points in the plane, and an integer \(K\) with \(1\le K\le n^2\), one can compute the \(K\)-th smallest distance among all interpoint distances between \(A\) and \(B\) in \(O(n^{4/3})\) deterministic time.
\end{theorem}


\section{The asymmetric case}
\label{sec:asymmetric}

In this section, we consider the asymmetric version of the distance
selection problem, where the sizes of the two point sets are not
necessarily equal. Let \(m=|A|\) and \(n=|B|\).

Our main result is the following.

\begin{theorem}
\label{thm:asy}
Given a set \(A\) of \(m\) points and a set \(B\) of \(n\) points in
the plane, and an integer \(K\) with \(1\le K\le mn\), one can compute
the \(K\)-th smallest distance among all interpoint distances between
\(A\) and \(B\) in \(O(n\log^2 m + m\log^2 n + m^{2/3}n^{2/3})\) deterministic time.
\end{theorem}

Most of the techniques developed for the symmetric case continue to
apply. We therefore only highlight the new ideas and refer to the
previous sections whenever the analysis is identical. However, certain care must be taken in order to achieve the claimed complexity. 
We use the same
notation as before, including \(S\), \(\mathcal D\), \(\delta^*\),
\(K\), \(K_\delta\), and \(\kappa(I)\).

Without loss of generality, assume \(m\le n\).
We distinguish two cases.

\begin{itemize}
\item If \(n\ge m^2\), we give an
      \(O(n\log^2 m)\)-time algorithm.

\item If \(n<m^2\), we give an algorithm running in
      \(O(n\log^2 m+m^{2/3}n^{2/3})\) time.
\end{itemize}

Combining the two cases yields
\(O(n\log^2 m+m\log^2 n+m^{2/3}n^{2/3})\) time for arbitrary
\(m\) and \(n\).

We begin with the easier case \(n\ge m^2\).

\subsection{The case \(\boldsymbol{n\ge m^2}\)}

Throughout this subsection, we assume \(n\ge m^2\).

We begin with the following exact counting lemma.

\begin{lemma}
\label{lem:exact-count}
Suppose that \(n\ge m^2\). Given an interval \(I\), one can compute
\(\kappa(I)\) in \(O(n\log m)\) time.
\end{lemma}

\begin{proof}
Let \(\mathcal D\) denote the set of annuli defined by the points of
\(A\) with respect to the interval \(I\).

Construct a hierarchical \((1/r)\)-cutting
\(\{\Xi_0,\Xi_1,\ldots,\Xi_k\}\) for \(\mathcal D\) with \(r=m\).
The construction takes \(O(mr)\) time.
Since \(r=m\), we can assume that every cell of the last cutting \(\Xi_k\) has an empty
conflict list. Consequently, no annulus of \(\mathcal D\) has a
boundary circle intersecting the interior of any cell of \(\Xi_k\).

For every cell \(\sigma\in\Xi_k\), compute
\(B(\sigma)=B\cap\sigma\). This takes \(O(n\log r)\) time over all
cells.

Next, compute for every cell \(\sigma\in\Xi_k\) the number
\(m_\sigma\) of annuli of \(\mathcal D\) that completely contain
\(\sigma\). As in Section~\ref{subsec:quasi-large}, all values
\(m_\sigma\) can be computed in \(O(mr)\) time.

Notice that
\[
\kappa(I)
=
\sum_{\sigma\in\Xi_k}
m_\sigma\,|B(\sigma)|.
\]

Therefore, after computing \(m_\sigma\) and \(B(\sigma)\) for all cells $\sigma\in \Xi_k$, \(\kappa(I)\) can be readily obtained.
The total running time is
\(O(mr+n\log r)\), which is \(O(n\log m)\) since \(r=m\) and
\(n\ge m^2\).
\end{proof}

By Lemma~\ref{lem:exact-count}, given any value \(\delta\), we can
determine whether \(\delta<\delta^*\) in \(O(n\log m)\) time. 

We now describe the distance selection algorithm.

\paragraph{Overview.}
The algorithm maintains an interval \(I\) containing \(\delta^*\).
The distances in \(S\cap I\) are called the
\emph{candidate distances}. In each iteration, we shrink \(I\) so that
the number of candidate distances decreases by a constant factor.
Initially, \(I=(0,+\infty)\), and therefore
\(|S\cap I|=O(mn)\). After \(O(\log m)\) iterations, each taking
\(O(n\log m)\) time, we obtain an interval \(I\) satisfying
\(|S\cap I|=O(n)\) and containing \(\delta^*\).
The total time of all iterations is thus \(O(n\log^2 m)\).
Finally, we apply a distance-reporting step that enumerates all distances
in \(S\cap I\) and computes \(\delta^*\) in additional
\(O(n\log m)\) time.

\paragraph{A general \(\boldsymbol{j}\)-th iteration.}

Suppose we are given an interval \(I_{j-1}\).
The goal is to compute an interval \(I_j\) such that either

\begin{enumerate}
\item \(|S\cap I_j|=O(n)\), in which case we proceed to the
      distance-reporting step; or

\item \(|S\cap I_j|\le \rho\,|S\cap I_{j-1}|\) for some constant
      \(\rho<1\).
\end{enumerate}

We first compute \(\kappa(I_{j-1})\) using
Lemma~\ref{lem:exact-count}. If
\(\kappa(I_{j-1})\le c\, n\) for a sufficiently large constant \(c\),
then \(|S\cap I_{j-1}|=O(n)\), and we proceed directly to the
distance-reporting step.

Otherwise, we compute a full biclique partition by applying the
algorithm for Recurrence~\eqref{formula:PrimalRecurrence} with
\(r=m\). Since every cell of the last cutting has an empty conflict
list, the resulting collection
\(\Gamma=\{A_\sigma\times B(\sigma):\sigma\in\Xi\}\)
forms a full biclique partition of \(S\cap I_{j-1}\).
The construction takes \(O(mr+n\log r)\) time, which is
\(O(n\log m)\), and the total biclique size satisfies
$\sum_{A_t\times B_t\in\Gamma} (|A_t|+|B_t|) =O(n\log m)$.

We then construct the expander graphs for the bicliques of \(\Gamma\)
and form the weighted distance set \(E\) exactly as in
Section~\ref{subsec:algodes}. This also takes \(O(n\log m)\) time.
Next, we partition \(E\) into \(O(1)\) basic intervals and use
Lemma~\ref{lem:exact-count} to identify the unique interval containing
\(\delta^*\). This interval is returned as \(I_j\).
By the same analysis as in
Section~\ref{subsec:algodes} (see also \cite{ref:WangIm25,ref:KatzAn97}),
there exists a constant \(\rho<1\) such that
$|S\cap I_j| \le \rho\,|S\cap I_{j-1}|$.
Since Lemma~\ref{lem:exact-count} is invoked only \(O(1)\) times,
the total running time of the iteration is \(O(n\log m)\).

\paragraph{Distance-reporting.}
We are given an interval \(I\) containing \(\delta^*\) with
\(|S\cap I|=O(n)\). The goal is to finally compute $\delta^*$. 

Assume \(I=(\alpha,\beta]\). We first compute \(K_\alpha\) using
Lemma~\ref{lem:exact-count}. Then \(\delta^*\) is the
\((K-K_\alpha)\)-th smallest distance in \(S\cap I\).
To enumerate all distances in \(S\cap I\), we again compute a full
biclique partition \(\Gamma\) as above. This takes \(O(n\log m)\) time.
We then explicitly output all distances recorded in \(\Gamma\), which
requires \(O(|S\cap I|)=O(n)\) additional time since \(|S\cap I|=O(n)\).
Finally, we apply the linear-time selection algorithm to obtain
\(\delta^*\) in \(O(n)\) time.
The total running time of the distance-reporting step is therefore
\(O(n\log m)\).

\medskip

We conclude that if \(n\ge m^2\), then the distance selection problem
can be solved in \(O(n\log^2 m)\) time.

\paragraph{Remark.}
One may wonder whether the running time can be improved further to
\(O(n\log m)\), for example by incorporating approximate counting.
The main obstacle is point location. In each iteration, we must locate
all \(n\) points of \(B\) in a hierarchical cutting. In our symmetric case
algorithm, we were able to perform these point locations in constant
amortized time by using the \(\Gamma\)-framework on sufficiently small
subproblems. Here, however, the \(\Gamma\)-framework is no longer applicable. Consequently, each point location costs
\(O(\log m)\) time, forcing the running time of every iteration to be
\(O(n\log m)\).

\subsection{The case $\boldsymbol{n<m^2}$}

We now consider the case \(n<m^2\). As before, we follow a
prune-and-search framework and maintain an interval \(I\) containing
\(\delta^*\). Initially, \(I=(0,+\infty)\), and therefore
\(|S\cap I|=O(mn)\).

The algorithm consists of three main procedures.

\begin{enumerate}
\item Reduce the number of candidate distances from \(O(mn)\) to
      \(O(m^2)\).

\item Reduce the number of candidate distances from \(O(m^2)\) to
      \(O(m^{2/3}n^{2/3})\).

\item Enumerate the remaining candidate distances and compute
      \(\delta^*\).
\end{enumerate}

The first and third procedures are relatively straightforward. The main
technical challenge lies in the second procedure.

\paragraph{A standard subroutine.}

A recurring step in our algorithm is to apply Recurrence
\eqref{formula:PrimalRecurrence} with \(r=n/m\). We refer to this as
the \emph{standard subroutine}.
Substituting \(r=n/m\) into
Recurrence~\eqref{formula:PrimalRecurrence} yields
\begin{equation}
\label{equ:second10}
T(m,n)
=
O(n\log r)
+
O(r^2)\cdot T(q,q),
\qquad
q=\frac{m^2}{n}.
\end{equation}
The standard subroutine runs in \(O(n\log r)\) time and produces a
collection \(\Gamma=\{A_t\times B_t\}\) of bicliques whose total size
is \(O(n\log r)\).
Throughout the remainder of the section, we always use
$r=n/m$, $q=m^2/n$, and let \(N(\Gamma)\) denote the number of distances recorded by
\(\Gamma\).

Since \(m\le n<m^2\), we have \(r\le m\) and therefore
\(\log r=O(\log m)\). Consequently, the running time of the standard
subroutine is \(O(n\log m)\), which is the bound we will use in the following discussions.

The significance of Recurrence~\eqref{equ:second10} is that every recursive
subproblem has the form \(T(q,q)\). In other words, after one
application of the recurrence, the problem becomes symmetric.
Therefore, the tools developed in the previous sections (e.g., 
Theorems~\ref{thm:approx-count} and
\ref{thm:distance-representation}) can be applied directly to $T(q,q)$.

It is also useful to observe that
$r^2q^{4/3}
=
m^{2/3}n^{2/3}$,
which is precisely the complexity term appearing in
Theorem~\ref{thm:asy}.

\subsubsection{The first main procedure}

We begin with the following approximate counting lemma.

\begin{lemma}[Approximate Counting]
\label{lem:first-approx-count}
For any interval \(I\), one can compute in \(O(n\log m)\) time an
integer \(\hat\kappa(I)\) satisfying
$|\kappa(I)-\hat\kappa(I)|
\le
c_0\,m^2$
for some constant \(c_0>0\).
\end{lemma}

\begin{proof}
We invoke the standard subroutine and return
$\hat\kappa(I)
=
\sum_{A_t\times B_t\in\Gamma}
|A_t||B_t|$.
The running time is \(O(n\log m)\).

Clearly, $|\kappa(I)-\hat\kappa(I)|$ is bounded by the number of uncertain pairs recorded in the recursive subproblems of $T(q,q)$. 
There are \(O(r^2)\) such subproblems, each contributing at most \(O(q^2)\) pairs. Hence
$|\kappa(I)-\hat\kappa(I)|
=
O(r^2q^2)
=
O(m^2)$,
since \(r=n/m\) and \(q=m^2/n\).
\end{proof}

Lemma~\ref{lem:first-approx-count} immediately yields the following
quasi-decision lemma, whose proof is similar to that of
Theorem~\ref{thm:quasidecision}.

\begin{lemma}[Quasi-Decision]
\label{lem:first-quasidecision}
Given an interval \(I\) and an integer \(K\) with \(1\le K\le mn\),
there exists an \(O(n\log m)\)-time algorithm whose output has one of
the following three outcomes:
(1) the algorithm reports \(\kappa(I)<K\); (2) the algorithm reports \(\kappa(I)>K\); (3) the algorithm reports ``unsure'', in which case
      $|\kappa(I)-K|      \le       c_0\,m^2$           for some constant \(c_0>0\).
\end{lemma}

We now describe the first main procedure. Its goal is to reduce the
number of candidate distances from \(O(mn)\) to \(O(m^2)\) while
maintaining an interval \(I\) containing \(\delta^*\).

Initially, \(I=(0,+\infty)\), and therefore
\(|S\cap I|=O(mn)\). The algorithm proceeds in iterations. In each
iteration, which runs in $O(n\log m)$ time, the number of candidate distances is reduced by a constant
factor. After \(O(\log (n/m))\) iterations, we obtain an interval
\(I\) containing \(\delta^*\) such that \(|S\cap I|=O(m^2)\).

Since each iteration takes \(O(n\log m)\) time, the total running time
of the procedure is bounded by \(O(n\log m\log (n/m))\), which is $O(n\log^2 m)$ since $n<m^2$.

\paragraph{A general $\boldsymbol{j}$-th iteration.}

Suppose we are given an interval \(I_{j-1}\). The goal is to compute an
interval \(I_j\) such that either

\begin{enumerate}
\item \(|S\cap I_j|=O(m^2)\), in which case the first main procedure
      terminates; or

\item \(|S\cap I_j|\le \rho |S\cap I_{j-1}|\) for some constant
      \(\rho<1\).
\end{enumerate}

We begin by invoking the standard subroutine. The running time is
\(O(n\log m)\).
Recall that \(N(\Gamma)\) denote the number of pairs recorded by the bicliques
of \(\Gamma\). Let \(M\) denote the total number of uncertain pairs contained in the recursive subproblems \(T(q,q)\).
As shown in the proof of Lemma~\ref{lem:first-approx-count},
\(M=O(m^2)\). Thus we may assume that
$M\le c_1m^2$
for some constant \(c_1>0\).

If \(N(\Gamma)\le c_2m^2\) for a sufficiently large constant \(c_2\),
then $|S\cap I_{j-1}| = N(\Gamma)+M = O(m^2)$.
Hence, in this case, the first main procedure terminates and returns
\(I_{j-1}\).
Assume therefore that \(N(\Gamma)>c_2m^2\).

We next construct the expander graphs for the bicliques of \(\Gamma\)
and form a weighted distance set \(E\) exactly as in
Section~\ref{subsec:algodes}. Since the total size of the bicliques in
\(\Gamma\) is \(O(n\log m)\), this step takes \(O(n\log m)\) time.
We then partition \(E\) into \(O(1)\) basic intervals and use
Lemma~\ref{lem:first-quasidecision} in place of the quasi-decision
algorithm from Section~\ref{subsec:algodes}. This yields an interval
\(I_j\).

Using analysis similar to Section~\ref{subsec:algocor}, we can prove $|S\cap I_{j}|\leq \rho\cdot |S\cap I_{j-1}|$ for some constant $\rho<1$. Intuitively, we can make $c_2$ large enough so that $M$ and the term $c_0m^2$ in Lemma~\ref{lem:first-quasidecision} can be neglected compared to $c_2m^2$. For example, if we follow the same notation as in Section~\ref{subsec:algocor}, then we have $m_{j-1}=N(T)$ and $m'_{j-1}\leq M$. Since $M\leq c_1\,m^2$ and $N(\Gamma)>c_2\,m^2$, we have $m'_{j-1}\leq \frac{c_1}{c_2}\cdot m_{j-1}$. Since $n_{j-1}=m_{j-1}+m'_{j-1}$, it follows that  $m'_{j-1}\leq \frac{c_1}{c_1+c_2}\cdot n_{j-1}$. This inequality can be used to replace $m'_{j-1}\leq 0.6 n_{j-1}$ that was obtained by Theorem~\ref{thm:distance-representation} (note that here $\Gamma$ plays the role of Theorem~\ref{thm:distance-representation} and that is why we do not have an explicit distance representation theorem here). Similarly, we can prove $|S\cap \hat I|\leq 2c_0\,m^2$ as a counterpart to Lemma~\ref{lem:middleinterval}. 

Therefore, each iteration runs in \(O(n\log m)\) time and reduces the
number of candidate distances by a constant factor. After
\(O(\log(n/m))\) iterations, we obtain an interval \(I\) containing
\(\delta^*\) such that \(|S\cap I|=O(m^2)\).
The total running time of the first main procedure is
\(O(n\log^2 m)\).

\subsubsection{The second main procedure}

We now describe the second main procedure, which is the most
challenging part of the asymmetric distance selection algorithm.

Given an interval \(I\) containing \(\delta^*\) with
\(|S\cap I|=O(m^2)\), the goal is to shrink \(I\) so that
\(|S\cap I|=O(m^{2/3}n^{2/3})\) while maintaining
\(\delta^*\in I\).

We begin with the following approximate counting lemma.

\begin{lemma}[Approximate Counting]
\label{lem:second-approx-count}
Let \(I\) be an interval. For every integer \(i\) with
\(1\le i\le O(\log^* q)\), one can compute an integer
\(\hat\kappa(I)\) such that
\[
|\kappa(I)-\hat\kappa(I)|
\le
c_0^*\,m^{2/3}n^{2/3}(\log^{(i)}q)^4
\]
for some constant \(c_0^*>0\).
The running time is
\(O(n\log m+m^{2/3}n^{2/3}/\log^{(i)}q)\) for \(i=1,2\), and
\(O(n\log m+m^{2/3}n^{2/3}/(\log^{(i)}q)^2)\) for \(i\ge3\).
\end{lemma}

\begin{proof}
We only discuss the case \(i\ge3\); the cases \(i=1,2\) are analogous.

We first invoke the standard subroutine in \(O(n\log m)\) time and add
the number of pairs recorded by the resulting bicliques \(\Gamma\) to
\(\hat\kappa(I)\).

For each recursive subproblem \(T(q,q)\) in Recurrence~\eqref{equ:second10}, it is for a subset $A'\subseteq A$ and $B'\subseteq B$ with $|A'|\leq q$ and $|B'|\leq q$. Let \(\kappa'(I)\) denote the number of distances between
\(A'\) and \(B'\) lying in \(I\). 
Apply Theorem~\ref{thm:approx-count} to \((A',B')\) with $i$ and $I$ yields a number $\hat \kappa'(I)$ that approximates $\kappa'(I)$ with error bounded by $c_0\,q^{4/3}(\log^{(i)}q)^4$. The runtime is $O(q^{4/3}/(\log^{(i)}q)^2)$. 
We add $\hat\kappa'(I)$ to $\hat\kappa(I)$. 

Since there are \(O(r^2)\) subproblems, the total error is bounded by
\[
O(r^2)\cdot c_0q^{4/3}(\log^{(i)}q)^4
=
O(m^{2/3}n^{2/3}(\log^{(i)}q)^4),
\]

where we use the identity
\(r^2q^{4/3}=m^{2/3}n^{2/3}\).

Similarly, the total time spent on all subproblems is

\[
O(r^2)\cdot
O\!\left(
\frac{q^{4/3}}{(\log^{(i)}q)^2}
\right)
=
O\!\left(
\frac{m^{2/3}n^{2/3}}
     {(\log^{(i)}q)^2}
\right).
\]

Adding the \(O(n\log m)\) time of the standard subroutine completes the
proof.
\end{proof}

Lemma~\ref{lem:second-approx-count} immediately yields the following
quasi-decision lemma, whose proof is similar to that of
Theorem~\ref{thm:quasidecision}.

\begin{lemma}[Quasi-Decision]
\label{lem:second-quasidecision}
Given an interval \(I\) and an integer \(K\) with
\(1\le K\le mn\), for every \(i\) with
\(1\le i\le O(\log^*q)\), there exists an algorithm whose output has
one of the following three outcomes:
(1) the algorithm reports \(\kappa(I)<K\); (2) the algorithm reports \(\kappa(I)>K\); (3) the algorithm reports ``unsure'', in which case
      $|\kappa(I)-K|      \le      c_0^*\,m^{2/3}n^{2/3}(\log^{(i)}q)^4$    for some constant \(c_0^*>0\).
The running time is
\(O(n\log m+m^{2/3}n^{2/3}/\log^{(i)}q)\) for \(i=1,2\), and
\(O(n\log m+m^{2/3}n^{2/3}/(\log^{(i)}q)^2)\) for \(i\ge3\).
\end{lemma}

We now describe the second main procedure. The framework is identical
to that of Section~\ref{sec:selection}.

Recall that \(|S\cap I|=O(m^2)\) and \(\delta^*\in I\).
The algorithm proceeds in phases.
For \(i=1,2,\ldots,O(\log^*q)\), the goal of the \(i\)-th phase is to
transform an interval \(I\) satisfying
\[
|S\cap I|
=
O(m^{2/3}n^{2/3}(\log^{(i-1)}q)^4)
\]
(and \(O(m^2)\) when \(i=1\)) into a smaller interval
\(I'\subseteq I\) satisfying
\[
|S\cap I'|
=
O(m^{2/3}n^{2/3}(\log^{(i)}q)^4)
\]
while maintaining \(\delta^*\in I'\).

Initially, \(|S\cap I|=O(m^2)\). After \(O(\log^*q)\) phases, we
obtain an interval \(I\) satisfying
\(|S\cap I|=O(m^{2/3}n^{2/3})\), which is precisely the goal of the
second main procedure.
Each phase consists of \(O(\log^{(i)}q)\) iterations. Each iteration
runs in
\(O(n\log m+m^{2/3}n^{2/3}/\log^{(i)}q)\) time for \(i=1,2\), and in
\(O(n\log m+m^{2/3}n^{2/3}/(\log^{(i)}q)^2)\) time for \(i\ge3\).
Consequently, the total running time over all phases is
\[
O(n\log m)\cdot 
\sum_{i=1}^{O(\log^*q)}
\log^{(i)}q
+
O(m^{2/3}n^{2/3})\cdot 
\left(
1+
\sum_{i=3}^{O(\log^*q)}
\frac1{\log^{(i)}q}
\right),
\]
which is
\[
O(n\log m\log q)
+
O(m^{2/3}n^{2/3})
=
O(n\log^2 m+m^{2/3}n^{2/3}),
\]
since \(q=m^2/n\le m\).

\paragraph{Algorithm description.}
We describe a general \(i\)-th phase.
Before the phase begins, we are given an interval \(I\) satisfying
$|S\cap I|
=
O(m^{2/3}n^{2/3}(\log^{(i-1)}q)^4)$
and containing \(\delta^*\). The goal is to compute a subinterval
\(I'\subseteq I\) such that
$|S\cap I'|
=
O(m^{2/3}n^{2/3}(\log^{(i)}q)^4)$
and \(\delta^*\in I'\).

As in Section~\ref{sec:selection}, the phase proceeds in iterations.
At the beginning of the \(j\)-th iteration, we are given an interval
\(I_{j-1}\) containing \(\delta^*\), with \(I_0=I\).
Each iteration either

\begin{enumerate}
\item produces a smaller interval \(I_j\subseteq I_{j-1}\) such that
      \(\delta^*\in I_j\) and
      \(|S\cap I_j|\le \rho |S\cap I_{j-1}|\) for some constant
      \(\rho<1\); or

\item produces a smaller interval \(I_j\subseteq I_{j-1}\) such that
      \(\delta^*\in I_j\) and
      \(|S\cap I_j|
      \le
      c^*m^{2/3}n^{2/3}(\log^{(i)}q)^4\)
      for a sufficiently large constant \(c^*\).
\end{enumerate}

Since the phase starts with
\(|S\cap I|
=
O(m^{2/3}n^{2/3}(\log^{(i-1)}q)^4)\),
the first outcome can occur only \(O(\log^{(i)}q)\) times before the
second outcome must occur.\footnote{Note that for \(i=1\), the number of
candidate distances decreases from \(O(m^2)\) to
\(O(m^{2/3}n^{2/3}(\log q)^4)\). Since
\(m^2/(m^{2/3}n^{2/3}(\log q)^4)
=
q^{2/3}/(\log q)^4\),
this requires only \(O(\log q)\) iterations.}
Unlike the symmetric case algorithm, we do not have an efficient test for
whether the second outcome has already occurred. Instead, we simply run
the phase for \(\hat c\log^{(i)}q\) iterations, where \(\hat c\) is a
sufficiently large constant depending on \(\rho\) and \(c^*\).

In the following, we describe a single iteration. For simplicity, we
assume \(i\ge3\). The cases \(i=1,2\) are identical except for the
running-time bounds.
We assume that 
\[
|S\cap I_{j-1}| \ge c^*m^{2/3}n^{2/3}(\log^{(i)}q)^4,
\] 
since otherwise the second outcome already holds.

We first invoke the standard subroutine in \(O(n\log m)\) time,
obtaining a collection \(\Gamma\) of bicliques.
We use index $p$ to refer to a recursive subproblem $T(q,q)$ in Recurrence~\eqref{equ:second10}. 
For each subproblem \(p\), it is associated with a subset $A_p\subseteq A$ and $B_p\subseteq B$ with $|A_p|\leq q$ and $|B_p|\leq q$. Let \(S_p\) denote the set of interpoint distances between
\(A_p\) and \(B_p\). Depending on the context, \(S_p\) may also refer
to the corresponding point pairs.

Let $S_0$ denote the set of distances of $S$ that are not in $S_p$ for any subproblem $p$. 
Then \(S_0\) and the sets \(S_p\) for all subproblems $p$ form a partition of
\(S\). Moreover, \(S_0\cap I_{j-1}\) is precisely the set of distances
recorded by \(\Gamma\).

For every subproblem \(p\), we apply Theorem~\ref{thm:distance-representation} to \((A_p,B_p)\) and \(I_{j-1}\) with parameter \(i\).
This produces a collection of bicliques \(\Gamma_p\) in \(O(q^{4/3}/(\log^{(i)}q)^2)\) time. Since there are \(O(r^2)\) subproblems, the total running time is $O(r^2)\cdot O({q^{4/3}}/{(\log^{(i)}q)^2})=O({m^{2/3}n^{2/3}}/{(\log^{(i)}q)^2})$.

Let \(N(\Gamma_p)\) denote the number of pairs recorded by
\(\Gamma_p\).

We call a subproblem \(p\) \emph{bounded} if
$N(\Gamma_p)
\le
c_1q^{4/3}(\log^{(i)}q)^4$,
where \(c_1\) is the constant from
Theorem~\ref{thm:distance-representation}; otherwise \(p\) is
\emph{unbounded}.
Let \(P_1\) and \(P_2\) denote the sets of bounded and unbounded
subproblems, respectively.
Define
\(S_1=\bigcup_{p\in P_1}S_p\),
\(S_2=\bigcup_{p\in P_2}S_p\), and
\(S_{02}=S_0\cup S_2\).
Then \(S_{02}\) and \(S_1\) form a partition of \(S\).

For every bounded subproblem \(p\in P_1\),
Lemma~\ref{lem:bounded} implies that
\(|S_p\cap I_{j-1}|
=
O(q^{4/3}(\log^{(i)}q)^4)\).
Since there are \(O(r^2)\) bounded subproblems,
$|S_1\cap I_{j-1}|
=
O(m^{2/3}n^{2/3}(\log^{(i)}q)^4)$.
Hence, for some constant \(c'>0\),
\[
|S_1\cap I_{j-1}|
\le
c'm^{2/3}n^{2/3}(\log^{(i)}q)^4.
\]
Combining this with the assumption on
\(|S\cap I_{j-1}|\) yields
\begin{equation}
\label{equ:second20}
|S_1\cap I_{j-1}|
\le
\frac{c'}{c^*}|S\cap I_{j-1}|.
\end{equation}

Observe that \(S_1\) plays exactly the role of the distances of the unresolved pairs in
the symmetric analysis.

For the bicliques in \(\Gamma\) and in all collections
\(\Gamma_p\) with \(p\in P_2\), we construct expander graphs and form a
weighted distance set \(E\) exactly as in
Section~\ref{subsec:algodes}.
The total size of the bicliques in \(\Gamma\) is \(O(n\log m)\), while
the total size of all bicliques in \(\Gamma_p\) over
all \(p\in P_2\) is
$O(r^2)\cdot O(q^{4/3}/(\log^{(i)}q)^2)=O(m^{2/3}n^{2/3}/(\log^{(i)}q)^2)$.
Hence, constructing all expander graphs and the set \(E\) takes
$O(n\log m + m^{2/3}n^{2/3}/(\log^{(i)}q)^2)$ time.

Finally, we test the values of \(E\) exactly as in
Section~\ref{subsec:algodes}, using
Lemma~\ref{lem:second-quasidecision} in place of
Theorem~\ref{thm:quasidecision}. This yields an interval \(I_j\).

Since only \(O(1)\) quasi-decision calls are required, the total
running time of the iteration remains
$O(n\log m+m^{2/3}n^{2/3}/(\log^{(i)}q)^2)$.


\paragraph{Algorithm Correctness.}

Using an argument analogous to the proof of
Lemma~\ref{lem:prune}, one can show that
\begin{equation}
\label{equ:second30}
|S_{02}\cap I_j|
\le
\rho'|S_{02}\cap I_{j-1}|
\end{equation}

for some constant \(\rho'<1\). The only additional ingredient is that,
for every unbounded subproblem \(p\in P_2\), by definition, $N(\Gamma_p)> c_1q^{4/3}(\log^{(i)}q)^4$. 
Since $|S_p\cap I_{j-1}|\geq N(\Gamma_p)$, we have $|S_p\cap I_{j-1}| > c_1q^{4/3}(\log^{(i)}q)^4$. Consequently, 
$N(\Gamma_p)
\ge
0.4\,|S_p\cap I_{j-1}|$
by Theorem~\ref{thm:distance-representation}.

We now prove that
\[
|S\cap I_j|
\le
\rho\,|S\cap I_{j-1}|
\]
for some constant \(\rho<1\).

Since \(S_{02}\) and \(S_1\) form a partition of \(S\),
\[
|S\cap I_j|
=
|S_1\cap I_j|
+
|S_{02}\cap I_j|.
\]

Because \(I_j\subseteq I_{j-1}\), we have
\(S_1\cap I_j\subseteq S_1\cap I_{j-1}\).
Therefore,
\[
|S\cap I_j|
\le
|S_1\cap I_{j-1}|
+
|S_{02}\cap I_j|.
\]

Combining \eqref{equ:second20} and \eqref{equ:second30} yields

\[
|S\cap I_j|
\le
\frac{c'}{c^*}|S\cap I_{j-1}|
+
\rho'|S_{02}\cap I_{j-1}|
\le
\left(
\frac{c'}{c^*}
+
\rho'
\right)
|S\cap I_{j-1}|.
\]

Since \(\rho'<1\), choosing \(c^*\) sufficiently large and defining
$\rho = \frac{c'}{c^*} +\rho'$
gives \(\rho<1\). This proves the first outcome of the iteration.

The above analysis assumes that
$|S\cap I_{j-1}|
\ge
c^*m^{2/3}n^{2/3}(\log^{(i)}q)^4$.
Otherwise, we simply run the same algorithm, which obtains an interval
\(I_j\subseteq I_{j-1}\) containing \(\delta^*\).
Since
$|S\cap I_j|
\le
|S\cap I_{j-1}|
\le
c^*m^{2/3}n^{2/3}(\log^{(i)}q)^4$,
the second outcome holds automatically.

Therefore, in each iteration, we simply run the same algorithm as above. Every iteration either reduces the number of candidate
distances by a constant factor or reaches the target bound
\(O(m^{2/3}n^{2/3}(\log^{(i)}q)^4)\).

\paragraph{Summary.}

We conclude that, within
\(O(n\log^2 m+m^{2/3}n^{2/3})\) time, we can compute an interval
\(I\) containing \(\delta^*\) such that
$|S\cap I|
=
O(m^{2/3}n^{2/3})$.

\subsubsection{The distance-reporting procedure}

We are given an interval \(I\) containing \(\delta^*\) such that
$|S\cap I|
=
O(m^{2/3}n^{2/3})$.
The goal of this procedure is to compute \(\delta^*\).

Assume that \(I=(\alpha,\beta]\). We first compute \(K_\alpha\), the
number of distances in \(S\) that are at most \(\alpha\). This can be
done in
\(O(m\log n+n\log m+m^{2/3}n^{2/3})\) time~\cite{ref:WangUn23}. Consequently, \(\delta^*\) is the
\((K-K_\alpha)\)-th smallest distance in \(S\cap I\).
Thus it suffices to enumerate all distances in \(S\cap I\) and then
apply the linear-time selection algorithm. The latter step takes
\(O(m^{2/3}n^{2/3})\) time since
\(|S\cap I|=O(m^{2/3}n^{2/3})\).

The remaining task is therefore to enumerate all distances in
\(S\cap I\). For this, we have the following more general result. 

\begin{lemma}
Given any interval \(I\), all distances in \(S\cap I\) can be
enumerated in
\(O(n\log m+m^{2/3}n^{2/3}+|S\cap I|)\) time.
\end{lemma}

\begin{proof}
We first invoke the standard subroutine in \(O(n\log m)\) time,
obtaining a collection \(\Gamma\) of bicliques. All distances recorded
by \(\Gamma\) are then reported explicitly, requiring constant time per
output distance.

For each recursive subproblem \(T(q,q)\) in Recurrence~\eqref{equ:second10},
it is associated with a subset $A'\subseteq A$ and $B'\subseteq B$ with $|A'|\leq q$ and $|B'|\leq q$.
By Corollary~\ref{cor:outputdistance}, all distances between
\(A'\) and \(B'\) lying in \(I\) can be enumerated in
\(O(q^{4/3})\) time, in addition to constant time per output distance.
Since there are \(O(r^2)\) subproblems, the total time spent on all
subproblems is
$O(r^2q^{4/3}) = O(m^{2/3}n^{2/3})$,
again in addition to constant time per output distance.

Combining the contributions of \(\Gamma\) and all recursive
subproblems proves the lemma.
\end{proof}

Applying the above lemma and using the bound
\(|S\cap I|=O(m^{2/3}n^{2/3})\), all distances in \(S\cap I\) can be
enumerated in
\(O(n\log m+m^{2/3}n^{2/3})\) time. A final linear-time selection then
computes \(\delta^*\) within the same asymptotic bound.

\subsubsection{Wrapping things up}

We have shown that, for the case \(n<m^2\), the first main procedure
requires \(O(n\log^2 m)\) time, the second main procedure requires
\(O(n\log^2 m+m^{2/3}n^{2/3})\) time, and the distance-reporting
procedure requires \(O(n\log m+m^{2/3}n^{2/3})\) time.
Therefore, the overall running time is
$O(n\log^2 m+m^{2/3}n^{2/3})$.
This establishes Theorem~\ref{thm:asy} for the case \(m\le n<m^2\).

Combining this with the case \(n\ge m^2\), we conclude that if
\(m\le n\), then the distance selection problem can be solved in
$O(n\log^2 m+m^{2/3}n^{2/3})$
time.

If \(m>n\), the roles of \(A\) and \(B\) can be exchanged, yielding a
running time of
$O(m\log^2 n+m^{2/3}n^{2/3})$.

Combining the two bounds proves Theorem~\ref{thm:asy}.

\bibliographystyle{plainurl}
\bibliography{reference}

@STRING{STOC= {ACM Symposium on Theory of Computing}}

@STRING{ESA= {European Symposia on Algorithms}}

@STRING{ALGORITHMICA="Algorithmica"}

@inproceedings{ref:WangAn26,
    author= "Haitao Wang",
    title= "An Optimal Algorithm for Computing Many Faces in Line Arrangements",
    booktitle="Proceedings of the 42nd International Symposium on Computational Geometry (SoCG)",
    year= "2026",
    pages= "95:1--95:14",
    doi = "10.4230/LIPIcs.SoCG.2026.95",
}

@inproceedings{ref:ChanFa25,
  author       = {Timothy M. Chan and
                  Zhengcheng Huang},
  title        = {Faster Algorithms for Reverse Shortest Path in Unit-Disk Graphs and
                  Related Geometric Optimization Problems: Improving the Shrink-And-Bifurcate
                  Technique},
  booktitle    = {Proceedings of the 41st International Symposium on Computational Geometry (SoCG)},  
  pages        = {32:1--32:14},  
  year         = {2025},
  doi          = {10.4230/LIPICS.SOCG.2025.32},  
}

@article{ref:WangCo23,
    author= "Haitao Wang",
    title = "Constructing Many Faces in Arrangements of Lines and Segments",
    journal= "Journal of Computational Geometry",
    volume=  "14",
    number= "",
    pages= "287--308",
    year= "2023",
    doi = "10.20382/jocg.v14i1a11",
}

@article{ref:WangUn23,
    author= "Haitao Wang",
    title = "Unit-disk range searching and applications",
    journal= "Journal of Computational Geometry",
    volume=  "14",
    number= "",
    pages= "343--394",
    year= "2023",
    doi = "10.20382/jocg.v14i1a13",
}

@article{ref:WangRe23,
    author= "Haitao Wang and Yiming Zhao",
    title = "Reverse Shortest Path Problem for Unit-Disk Graphs",
    journal= "Journal of Computational Geometry",
    volume=  "14",
    number= "",
    pages= "14--47",
    year= "2023",
    doi = "10.20382/jocg.v14i1a2",
}

@inproceedings{ref:KaplanTh23,
    author= "Haim Kaplan and Matthew J. Katz and Rachel Saban and Micha Sharir",
    title= "The Unweighted and Weighted Reverse Shortest Path Problem for Disk Graphs",
    booktitle="Proceedings of the 31st Annual European Symposium on Algorithms (ESA)",
    year= "2023",
    pages= "67:1--67:14",
    doi = "10.4230/LIPIcs.ESA.2023.67",
}

@article{ref:BespamyatnikhCo03,
  author    = {Sergei Bespamyatnikh},
  title     = {Computing Closest Points for Segments},
  journal   = {International Journal of Computational Geometry and Application},
  volume    = {13},
  number    = {},
  pages     = {419--438},
  year      = {2003},
  doi = {10.1142/S0218195903001268},
}

@article{ref:ChanHo23,
    author= "Timothy M. Chan and Da Wei Zheng",
    title = "Hopcroft's Problem, Log-Star Shaving, {2D} Fractional Cascading, and Decision Trees",
    journal= "ACM Transactions on Algorithms",
    volume=  "20",
    number= "3",
    pages= "24:1--24:27",
    year= "2023",
    doi = "10.1145/3591357",
}

@article{ref:WangIm25,
  author    = {Haitao Wang and Yiming Zhao},
  title     = {Improved algorithms for distance selection and related problems},
  journal   = {Algorithmica},
  volume    = {87},
  number    = {},
  pages     = {908--929},
  year      = {2025},
  doi = {10.1007/s00453-025-01305-z},
}

@article{ref:AgarwalCo98,
    author= "Pankaj K. Agarwal and Ji\v{r}\'{i} Matou\v{s}ek and Otfried Schwarzkopf",
    title= "Computing Many Faces in Arrangements of Lines and Segments",
    journal="SIAM Journal on Computing",
    year= "1998",
    pages= "491--505",
    volume="27",
    doi= "10.1137/S009753979426616X",
}

@article{ref:EdelsbrunnerTh90,
    author= "Herbert Edelsbrunner and Leonidas J. Guibas and Micha Sharir",
    title= "The Complexity and Construction of Many Faces in Arrangement of Lines and of Segments",
    journal="Discrete and Computational Geometry",
    year= "1990",
    pages= "161--196",
    volume= "5",
    doi = "10.1007/BF02187784",
}

@article{ref:ClarksonCo90,
    author= "Kenneth L. Clarkson and Herbert Edelsbrunner and Leonidas J. Guibas and Micha Sharir and Emo Welzl",
    title= "Combinatorial Complexity Bounds for Arrangement of Curves and Spheres",
    journal="Discrete and Computational Geometry",
    year= "1990",
    pages= "99--160",
    volume="5",
    doi= "10.1007/BF02187783",
}

@Article{ref:AgarwalPa902,
    author= "Pankaj K. Agarwal",
    title = "Partitioning Arrangements of Lines {II}: Applications",
    journal= "Discrete and Computational Geometry",
    volume=  "5",
    number= "",
    pages= "533--573",
    year= "1990",
    doi="10.1007/BF02187809",
}

@article{ref:EricksonNe96,
    author= "Jeff Erickson",
    title= "New Lower Bounds for {Hopcroft's} Problem",
    journal="Discrete and Computational Geometry",
    volume="16",
    year= "1996",
    pages= "389--418",
    doi= "10.1007/BF02712875",
}

@article{ref:AgarwalSe93,
    author= "Pankaj K. Agarwal and Boris Aronov and Micha Sharir and Subhash Suri",
    title= "Selecting distances in the plane",
    journal="Algorithmica",
    volume="9",
    year= "1993",
    pages= "495--514",
    doi = "10.1007/BF01187037", 
}

@inproceedings{ref:GoodrichGe93,
    author= "Michael T. Goodrich",
    title= "Geometric partitioning made easier, even in parallel",
    booktitle="Proceedings of the 9st Annual Symposium on Computational Geometry (SoCG)",
    year= "1993",
    pages= "73--82",
    doi= "10.1145/160985.161002",
}

@inproceedings{ref:ChazelleNe85,
    author= "Bernard Chazelle",
    title= "New techniques for computing order statistics in {Euclidean} space",
    booktitle="Proceedings of the 1st Annual Symposium on Computational Geometry (SoCG)",
    year= "1985",
    pages= "125--134",
    doi= "10.1145/323233.323251"
}

@article{ref:ChanOn01,
    author= "Timothy M. Chan",
    title= "On Enumerating and Selecting Distances",
    journal="International Journal of Computational Geometry and Application",
    volume="11",
    year= "2001",
    pages= "291--304",
    doi= "10.1142/S0218195901000511"
}

@article{ref:FloydEx75,
    author= "Robert W. Floyd and Ronald L. Rivest",
    title= "Expected time bounds for selection",
    journal="Communications of the ACM",
    volume="18",
    year= "1975",
    pages= "165--172",
    doi= "10.1145/360680.360691",
}

@article{ref:AvrahamTh15,
  title={The discrete and semicontinuous {F}r{\'e}chet distance with shortcuts via approximate distance counting and selection},
  author={Rinat B. Avraham and Omrit Filtser and Haim Kaplan and Matthew J. Katz and Micha Sharir},
  journal={ACM Transactions on Algorithms},
  volume={11},
  number={4},
  pages={Article No. 29},
  year={2015},
  doi={10.1145/2700222},
}

@inproceedings{ref:KatzNe24,
  author       = {Matthew J. Katz and
                  Rachel Saban and
                  Micha Sharir},
  title        = {Near-Linear Algorithms for Visibility Graphs over a 1.5-Dimensional
                  Terrain},
  booktitle    = {Proceedings of the 32nd Annual European Symposium on Algorithms (ESA)},  
  pages        = {77:1--77:17},
  year         = {2024},
  doi          = {10.4230/LIPICS.ESA.2024.77},  
}

@article{ref:AgarwalOn24,
  author       = {Pankaj K. Agarwal and
                  Matthew J. Katz and
                  Micha Sharir},
  title        = {On reverse shortest paths in geometric proximity graphs},
  journal      = {Computational Geometry: Theory and Applications},
  volume       = {117},
  pages        = {102053},
  year         = {2024},  
  doi          = {10.1016/J.COMGEO.2023.102053},
}

@Article{ref:ChazelleCu93,
    author= "Bernard Chazelle",
    title = "Cutting hyperplanes for divide-and-conquer",
    journal= "Discrete and Computational Geometry",
    volume=  "9",
    number= "2",
    pages= "145--158",
    year= "1993",
    doi="10.1007/BF02189314",
}

@Article{ref:KatzAn97,
    author= "Matthew J. Katz and Micha Sharir",
    title = "An Expander-Based Approach to Geometric Optimization",
    journal= "SIAM Journal on Computing",
    volume=  "26",
    number= "5",
    pages= "1384--1408",
    year= "1997",
    doi = "10.1137/S0097539794268649",
}

@Inproceedings{ref:LubotzkyEx86,
    author= "Alexander Lubotzky and Ralph Phillips and Peter Sarnak:",
    title= "Explicit Expanders and the {Ramanujan} Conjectures",
    booktitle= "Proceedings of the 18th Annual ACM Symposium on Theory of
		Computing (STOC)",
    year= "1986",
    pages= "240--246",
    doi= "10.1145/12130.12154",
}

@article{ref:WangAl26,
    author= "Haitao Wang",
    title = "Algorithms for computing closest points for segments",
    journal= "Computational Geometry: Theory and Applications",
    volume=  "130",
    number= "102196",
    pages= "1--20",
    year= "2026",
    doi= "10.1016/j.comgeo.2025.102196",
}

@article{ref:MatousekRa91,
    author= "Ji\u{r}\'{i} Matou\v{s}ek",
    title = "Randomized optimal algorithm for slope selection",
    journal= "Information Processing Letters",
    volume=  "39",
    number= "",
    pages= "183--187",
    year= "1991",
    doi= "10.1016/0020-0190(91)90177-J",
}

@Article{ref:BlumTi73,
    author= "Manuel Blum and Robert W. Floyd and Vaughan Pratt and Ronald L. Rivest and Robert E. Tarjan",
    title = "Time bounds for selection",
    journal= "Journal of Computer and System Sciences",
    volume=  "7",
    number= "",
    pages= "448--461",
    year= "1973",
    doi= "10.1016/S0022-0000(73)80033-9"
}

\end{document}